\documentclass[11pt]{article}

\usepackage{setspace}

\usepackage{url}
\usepackage{amssymb,amsmath,amsfonts}
\usepackage{graphicx,color,enumitem}
\usepackage{eurosym}
\usepackage{mathdots,mathtools}
\usepackage{calrsfs}	
\usepackage{amsthm}	
\usepackage[round]{natbib}
\usepackage[titletoc,title]{appendix}
\usepackage{titlesec}
\usepackage{subfig}
\usepackage{titling}
\usepackage{booktabs}

\allowdisplaybreaks

\definecolor{labelkey}{rgb}{0.6,0,1}

\definecolor{red}{rgb}{0,0.0,0}

 \usepackage[margin=1in]{geometry}

\newcommand{\E}{{\mathbb E}}

\DeclareMathOperator{\Var}{Var}
\newcommand{\nv}{\boldsymbol}

\newcommand{\Id}{{\boldsymbol{I}}}
\newcommand{\bbr}{\mathbb{R}}

\DeclareMathOperator{\argmin}{argmin}
\DeclareMathOperator{\Diag}{Diag}

\newtheorem{theorem}{Theorem}

\newtheorem{assumption}{Assumption}

\newtheorem{corollary}{Corollary}

\newtheorem{definition}{Definition}
\newtheorem{example}{Example}

\newtheorem{lemma}{Lemma}

\newtheorem{proposition}{Proposition}
\newtheorem{remark}[theorem]{Remark}

\numberwithin{equation}{section}
\numberwithin{theorem}{section}
\numberwithin{proposition}{section}
\numberwithin{lemma}{section}
\numberwithin{corollary}{section}

\definecolor{Red}{rgb}{0.00, 0.00, 0.00}
    \newcommand{\Red}{\color{Red}}
\definecolor{Blue}{rgb}{0.00, 0.00, 1.00}
    
\definecolor{Green}{rgb}{0.00, .30, 0.00}

\thanksmarkseries{alph}
    
\begin{document}

\title{Market Efficient Portfolios in a Systemic Economy\thanks{The content of an early version of this manuscript is included in the PhD dissertation of \cite{awiszus2020}.}}
\author{Kerstin Awiszus\thanks{Institute of Actuarial and Financial Mathematics \& House of Insurance, Leibniz Universit\"at Hannover, Germany. email: \texttt{kerstin.awiszus@insurance.uni-hannover.de, stefan.weber@insurance.uni-hannover.de}.} \hspace{1cm}  Agostino Capponi\thanks{Department of Industrial Engineering \& Operations Research, Columbia University, New York, USA. email: \texttt{ac3827@columbia.edu}.} \hspace{1cm} Stefan Weber\protect\footnotemark[2]  
}

\maketitle

\begin{abstract}
	\doublespacing
	We study the ex-ante minimization of market inefficiency, defined in terms of minimum deviation of market prices from fundamental values, from a centralized planner's perspective. Prices are pressured from exogenous trading actions of leverage targeting banks, which rebalance their portfolios in response to asset shocks. We characterize market inefficiency in terms of two key drivers, the banks' systemic significance and the statistical moments of asset shocks, and develop an explicit expression for the matrix of asset holdings which minimizes such inefficiency. Our analysis shows that to reduce inefficiencies, portfolio holdings should deviate more from a full diversification strategy if there is little heterogeneity in banks' systemic significance.
\end{abstract}
\vspace{0.2cm}

\doublespacing
\section{Introduction}\label{sec:intro}

Forced asset sales and purchases have been widely observed in financial markets. The most popular form of forced trading is that of fire sales, and has been extensively implemented by hedge funds and broker dealers during the global 2007-2009 financial crisis; see \cite{brunnermeier2008market} and \cite{ khandani2011happened} for empirical evidence. 
 
An asset is sold at a depressed price by a seller who faces financial constraints that become binding, i.e., when the seller becomes unable to pay his own creditors without liquidating the asset. For example, members of a clearinghouse need to post additional collateral if the value of their portfolios drops by a significant amount (\cite{pirrong2011economics}). Similarly, a mutual fund may need to liquidate assets at discounted prices if it faces heavy redemption requests from its investors, and does not have enough cash reserves at disposal (\cite{chen2010payoff}). Banks manage their leverage based on internal value at risk models (\cite{adrianShinRFS, ShinGS}), and may need to liquidate assets if negative shocks hit their balance sheets. 

Forced purchases, despite less emphasized, are also important in financial markets. For instance, empirical evidence (\cite{CovalStafford2007}) suggests that equity mutual funds substantially increase their existing positions if they experience large inflows, thus creating upward pressure in the price of stocks held by these funds. Such inflow-driven purchases produce trading opportunities for outsiders, who would be able to sell their assets and earn a significant premium.

Asset purchases and sales triggered by financial constraints  push asset prices away from fundamental values (\cite{shleifer1992liquidation}), a form of inefficiency that we analyze in a systemic economy. Typically, when a firm must sell assets to fulfill a financial constraint, the potential buyers with the highest valuation for the asset are other firms belonging to the same industry or investors with appropriate expertise. Those firms are likely to be in a similar financial situation, and thus unable to supply liquidity. The buyers of these assets are then outsiders, who value these assets less. A symmetric argument holds if the firm executes inflow-driven purchases.

When a firm impacts asset prices through its trading actions, other market participants who happen to hold the same assets on their balance sheets are also affected, and may in turn violate their financial constraints, making it necessary for them to take trading actions. Through this process, the trading risk becomes {\it systemic}, i.e, it imposes cascading effects on asset prices and  impacts the equity of market participants through common asset ownership.

We consider an economy consisting of leveraged institutions (henceforth, called \emph{banks}) that track a fixed leverage ratio. Empirically, this behavior has been well documented for commercial banks in the United States, see, e.g., \cite*{adrian2010liquidity}. After a shock hits an asset class, prices change and so does the bank's leverage ratio. To fulfill the financial constraint of targeting its leverage, the bank must then liquidate or purchase assets, depending on whether the experienced shock was positive or negative. Banks trade assets with other nonbanking institutions that we model collectively as a representative nonbanking sector, assumed to have a downward sloping demand function as in \cite{capponi2015price}. The equilibrium price of the asset is uniquely pinned down by the point at which the demand of the banking and nonbanking sector intersect.

We study market efficiency, measured by the mean squared deviation of fundamental market capitalization, where all assets are valued at the fundamental values, from market capitalization where assets are valued at market prices. The latter prices internalize the pressure imposed by trading activities, as banks leverage or deleverage in response to exogenous shocks to asset values. Clearly, the closer prices are to their fundamental values, the more efficient the market is. 

{The role of prices in aggregating information that is dispersed in the economy is discussed in the seminal work of \cite{Hayek45}.\footnote{In the context of secondary markets, this issue is, e.g., discussed in \cite{Leland92}, \cite{Dow97}, \cite{Subrahmanyam2001}, \cite{Dow2003}, and \cite{Goldstein08}.} The importance of market efficiency} can be micro-founded in terms of increasing the information content of prices, and as a result, better guiding the resource allocation in the economy. \cite{BrunnChina} provide a simple economic setting, in which maximizing social welfare is consistent with the objectives of minimizing the deviation of asset prices from fundamentals, and of reducing asset market volatility. Prices thereby facilitate the efficient allocation of scarce resources.

We develop an explicit characterization for the distribution of banks' holdings that ex-ante maximize market efficiency. We refer to those as the \emph{f-efficient holdings}.\footnote{The terminology f-efficient is used to emphasize that the notion of efficiency we consider is related to fundamentals.} The main insight resulting from our approach is the identification of {\Red a key driver}, the {\it systemic significance} vector, which captures the contribution of each bank to increased price pressures.\footnote{\Red We study market efficiency from the point of view of a centralized planner who chooses banks' holdings. These holdings may not be compatible with banks' incentives. We do not provide here a formal analysis of privately optimal banks' holdings and regulatory frameworks leading towards f-efficiency. Nevertheless, we elaborate  on this incentive alignment problem in Section~\ref{sec:empirical}.} 
The systemic significance depends on the banks' target leverage, the banks' trading strategies, and the illiquidity characteristics of the assets. We identify an ``aggregate first then allocate'' procedure to construct a solution to the quadratic minimization problem yielding f-efficient holdings. First, using the {probability law} of asset shocks, we construct a vector of auxiliary weighted holdings. Then, we distribute the holdings to the banks based on their contributions to market inefficiency, which is directly proportional to their systemic significance.

We show that portfolio diversification is f-efficient if asset price shocks are homogeneous and banks are heterogeneous in terms of their systemic significance. Our analysis suggests that as the shocks hitting an asset class become (statistically) larger, it is beneficial to transfer the holdings of such an asset from a more systemically significant to a less systemically significant bank to raise efficiency.
We demonstrate that the more homogeneous the economy is in terms of banks' systemic significance, the further away the matrix of f-efficient holdings is from the matrix of full asset diversification.

\paragraph{Literature Review} 

Existing literature has identified two main channels through which banks are interlinked. The first channel is through the liability side of the balance sheet. Banks have claims on their debtors, and once they are hit by shocks, they may become unable to honor their liabilities, potentially causing cascading effects through the system. Seminal contributions in this direction include \cite{EN01}, which provide an algorithm to measure contagion triggered by sequential defaults in the contractual network, and \cite{acemoglu2015systemic}, who analyze the stability of various network structures and their resilience to shocks of different sizes.\footnote{Other related works include \cite{EGJ14} and \cite{Gai}. \cite{GY14}, \cite{CapChenYao}, and \cite{RV13} account for the impact of bankruptcy costs at defaults in a counterparty network model of financial contagion. Measures of systemic downside risk are analyzed in the works by \cite*{CIM13}, \cite*{frw17}, and \cite*{bffm19}. }

The second channel is through the asset side of the balance sheet, as banks are interlinked through common portfolio holdings. Financial contagion arises when banks take hits on their balance sheets, typically because the price of their assets is subject to pressure due to forced purchases or sales (see also the discussion in the introduction). Our study contributes to this stream of literature, and is related to that of \cite{greenwood2015vulnerable}, who calibrate a model of fire-sale spillovers, assuming an economy of leverage targeting banks. Their work has been extended by \cite{capponi2015price}, who consider the higher order effects of fire-sales externalities in a similar leverage targeting model. \cite{duarte2018fire} construct and empirically valuate a measure of systemic risk generated by fire-sales externalities. The main components of their measure, namely banks' sizes, leverages, and illiquidity concentration, also constitute the primary determinants of the systemic significance vector in our model. Other works have considered models where contagion happens both through the asset and liability side of the balance sheet; see, for instance, the earlier work of \cite{cifuentes2005liquidity}, and the more recent works of  \cite{AFM13}, \cite{CLY14}, and \cite{weber2017joint}. 

The analytical infrastructure of our model builds on the work of \cite{greenwood2015vulnerable} and \cite{capponi2015price}. As in \cite{greenwood2015vulnerable}, we restrict attention to the first order effects of price pressures, i.e., those caused by the first round of banks' trading actions in response to shocks. As in \cite{capponi2015price}, price impact is determined by the capacity of the unconstrained nonbanking sector to absorb the trading pressure of the leverage constrained banks. While \cite{greenwood2015vulnerable}, \cite{duarte2018fire}, and \cite{capponi2015price} consider an ex-post model of asset contagion, where banks manage their assets after the shock has  occurred, the present paper conducts an \emph{ex-ante} analysis of balance sheet holdings. 

\cite{wagner2011systemic} analyzes the tradeoff between diversity on the systemic level and diversification at the banking level. While we focus on the centralized problem of maximizing market efficiency, \cite{wagner2011systemic} considers the privately optimal solution. 
{\Red \cite{Suff} study the resilience of a financial system with respect to asset sales by distressed institutions. They define a system to be resilient if the number of shares sold during the fire-sales process goes to zero as the shock becomes small. Unlike their study, we study the implications of purchases and sales on market efficiency.}\footnote{Portfolio similarity has also been considered in other industry sectors than banking. In the insurance industry, \cite{Mila} find a strong positive relationship between the portfolio similarity of insurance companies, and their quarterly common sales during the following year.}


Our work is also related to a branch of literature that has analyzed the stability of portfolio allocations, diversification, stress testing, and heavy tail risks of portfolios. \cite{contschaaning2017} develops a systemic stress testing model, and compares the asset pricing implications of threshold based versus target leveraging. 
Using a generalized branching process approach, \cite*{caccioli2014stability} identify a critical threshold for leverage which separates stable from unstable portfolio allocations (see also \cite{raffestin2014diversification}). In a risk-sharing context, \cite{ibragimov2011diversification} analyze the tradeoff of diversity and diversification for heavy-tailed risk portfolio distributions. 
\cite*{beale2011individual} analyze the individually and systemically optimal allocations in a simplified loss model consisting of a small number of banks and assets.



{\Red Our paper is also related to the existing literature on bubbles. \cite{BubblesCredit} decompose the price of firm stocks into a fundamental and a bubble component which capture deviations from fundamentals. There are both similarities and differences between our study and theirs. In both models, the emergence of an asset price bubble hinges on the critical assumption that firms are debt financed and do not issue new equity to finance investment.
However, in our study, all firms have full knowledge about the laws of asset values, and the key friction which generates deviations from fundamentals is the leverage targeting constraint. By contrast, in their model, the main financial frictions are credit constraints, and both lenders and borrowers have beliefs on the collateral value.  
Bubbles mitigate the credit constraint by allowing firms to borrow.\footnote{\Red If both the lending and the investing firm overestimate the value of collateral because of a bubble, the lending firm will be willing to lend more because it can monetize the bubble if a default occurs.} As a result, the firm can finance more investment and increase profits and the value of its assets. This positive feedback loop 
leads to the emergence of a stock price bubble in equilibrium.
\cite{NutzSchein} consider a continuous-time model of trading which allows for short-selling. Unlike our setup, where all banks agree on the distributions of fundamental prices, in their model agents have heterogeneous beliefs about the dynamics of the Markov state process that determines the asset's payoff. In equilibrium, optimists hold long positions with the intention of benefiting from the option of reselling later. By contrast, pessimists go short as they enjoy the option of delaying trading. In their paper, \cite{NutzSchein} measure the size of the bubble as the difference between the equilibrium price and the counterfactual price prevailing if no re-trading was allowed. Unlike our paper, the bubble is driven by disagreement on the evolution of the state process, rather than by the need of complying with prescribed leverage requirements.
}

\paragraph{Outline} The paper is organized as follows. Section \ref{sec:model} summarizes the asset price contagion model. Our main contributions start from Section \ref{sec:deviations}, where we define the quantitative measure of market efficiency, and identify key drivers of this measure such as the vector of banks' systemic significance. In Section \ref{sec:efficientholdings}, we characterize f-efficient allocations. 
Section \ref{sec:casestudies} provides case studies for a calibrated version of our model, which highlight the trade-off between diversification and diversity. {\Red In Section \ref{sec:empirical}, we discuss empirical and policy implications of our study.} We conclude in Section \ref{sec:conclusion}. We delegate proofs of results stated in the main body to an Electronic Companion~(Part \ref{app:proofs}). We present additional discussions and supporting material in sections \ref{app:spectralradius}-\ref{app:liquidationstrategies} of the Companion.


{\section{Model}\label{sec:model}}

To begin with, we introduce a few notations and definitions used throughout the paper. For two (column) vectors $u=(u_1,\ldots,u_n)^\top$ and $v=(v_1,\ldots,v_n)^\top$, we let $u\circ v=(u_1v_1,\ldots, u_n v_n)^\top$ denote the componentwise product. Similarly, $\frac{u}{v}=(u_1/v_1,\ldots,u_n/v_n)^\top$ denotes the componentwise ratio. We use $\Diag(u)$ to denote the diagonal matrix with vector $u$ on the diagonal. The identity matrix is denoted by $\nv I$, the vector or matrix of ones is denoted by $\nv 1$, and the vector or matrix of zeros is denoted by $\nv 0$, where the dimension is either specified explicitly, e.g., $\nv 1_K\in\mathbb{R}^{K}$, or clear from the context.

\begin{sloppypar}
Our analysis is developed within the one-period version of the price contagion model by \cite{capponi2015price}. We briefly review the essential elements here and provide a more extensive review of their model in E-Companion~\ref{CappLarsreview}. They consider a financial market consisting of two sectors: a banking sector with $N$ banks and a nonbanking sector. Each bank manages its asset portfolio to track a fixed leverage ratio, consistently with empirical evidence reported in the seminal contribution of \cite{adrian2010liquidity}. The nonbanking sector consists of institutions that do not engage in leverage targeting.
\end{sloppypar} 

There are $K$ types of assets available, whose market prices at time $t=0,1$ are denoted by $P^k_t$. We write
$P_t = (P^1_t , P^2_t , \ldots , P^K_t)^\top$
for the column vector of asset prices. Each asset $k$ is hit by an exogenous shock $Z^k$, which is modeled as a random variable. We use $Z = (Z^1, \ldots , Z^K)^\top$ to denote the vector of shocks. The vector of {\it fundamental values} of the assets at time $1$ is $P_0+Z$, while we use  $P_1$ to denote the vector of equilibrium prices which internalize banks' responses to shocks.
The quantity (number of units) of asset~$k$ held by bank $i$ at time $0$ is denoted by $Q^{ki}_0$. We use $Q^i_0 = (Q^{1i}_0 , Q^{2i}_0 , \ldots , Q^{Ki}_0)^\top \in\mathbb{R}^{K}$ to denote the vector of bank $i$'s holdings at $0$, and $\nv Q:=(Q_0^{ki})_{k=1,\dots,K,i=1,\ldots,N}\in\mathbb{R}^{K\times N}$ to denote the matrix of banks' holdings at time zero.  We use $\kappa^i$ to denote the leverage ratio (debt to equity ratio) targeted by bank $i$. Each bank $i$ executes an exogenous trading strategy $\alpha^{i}\in\bbr^K$ with $\sum_{k=1}^K\alpha^{ki}=1$, which specifies how a change in the amount of debt is offset by purchases or sales of the different assets in the portfolio.  
We denote by $Q^{k,\,\rm nb}_0$ the quantity of asset $k$ held by the nonbanking sector at time $0$. The market-clearing condition is given by
$Q^{\rm nb}_t + \sum\nolimits_{i=1}^N Q^i_t = Q_{\rm tot},$ $t = 0,1,$ where the vector $Q_{\rm tot}$ of aggregate supply is constant through time.

\cite{capponi2015price} introduce the systemicness matrix  
$
\nv S =  \sum_{i=1}^N\frac{\alpha^i}{\gamma\circ Q_0^{\rm nb}} \kappa^i Q_0^{i\top}  \in \bbr^{K\times K},
$
where $\gamma = (\gamma^1,\ldots,\gamma^K)$, and $\gamma^k$ is the illiquidity characteristic of asset~$k$. In componentwise form, we can write it as
$S^{k\ell} = \sum\nolimits_{i=1}^N \alpha^{ki} \tfrac{\kappa^i Q^{\ell i}_0}{\gamma^k Q^{k,\,\rm nb}_0},$ $k,\ell=1,\ldots,K.$
The systemicness matrix is the primary determinant of asset prices, and captures how a shock to asset $\ell$ propagates to asset $k$ through the banks' deleveraging activities. 

\begin{proposition} [Proposition 2.1. in \cite{capponi2015price}] \label{A:P:dyn}
	The changes in asset prices are given by
	$
	\Delta P = \left(\Id - \nv S \right)^{-1} Z  , $
	assuming that the matrix inverse exists. The invertibility of the matrix $\nv I-\nv S$, with $(\nv I-\nv S)^{-1}=\sum\nolimits_{j=0}^\infty \nv S^j$, is guaranteed by the following assumption. 
\end{proposition}
\begin{assumption}\label{ass:SmatrixInvertibleAndSpectralRadius} \footnote{We discuss Assumption \ref{ass:SmatrixInvertibleAndSpectralRadius}
	 in E-Companion \ref{app:spectralradius}.} We will always assume that the spectral radius of $\nv S$ is smaller than one. 
\end{assumption}


\section{Market Inefficiencies and Systemic Significance}\label{sec:deviations}
In this section, we introduce the measure used to quantify price deviation from fundamentals, and characterize the key quantities that determine market efficiency. Section~\ref{sec:marketcap} derives an explicit expression for the price deviation from fundamentals{, and states the objective function of minimizing market inefficiency in a systemic economy}. Section~\ref{sec:systsign} introduces a {\Red key quantity}, the banks' systemic significance, that quantifies the contribution to price pressures of each bank in the economy.

\subsection{Market Capitalization and Deviation from Efficiency}\label{sec:marketcap}
Banks actively manage their balance sheets in response to shocks, and this imposes a pressure on asset prices, pushing them away from fundamental values. The value of market capitalization at the end of period 1 is given by
$
MC^e:=Q_{\rm tot}^{\top} P_1 =Q_{\rm tot}^{\top} (P_0+\Delta P) = Q_{\rm tot}^{\top} \left(P_0+ \left(\Id - \nv S \right)^{-1}  Z\right),\label{eq:MC}
$
where we recall that $Z$ is the vector of exogenous price shocks (cf. Proposition \eqref{A:P:dyn}). In the absence of leverage-tracking banks, changes in prices are driven solely by changes in fundamentals, i.e., by the exogenous shocks. The resulting market capitalization is denoted by $MC^{f}$, and given by
$ MC^{f}:= Q_{\rm tot}^{\top} (P_0+Z).$
The contribution to realized asset prices arising from the presence of leverage targeting banks, denoted by $D:= MC^e-MC^{f}$, is a measure of market inefficiency.

Under the conditions discussed in Assumption~\ref{ass:SmatrixInvertibleAndSpectralRadius}, we obtain the first order approximation to market capitalization $MC^a:=Q_{\rm tot}^{\top} \left(P_0+  (\nv I+\nv S) Z \right),$
where we consider only the first term in the power series expansion of $(\nv I - \nv S )^{-1}$. The accuracy of this approximation depends on the spectral radius of the matrix $\nv S$, as shown in E-Companion \ref{app:accuracy}.

Recall that the $i$'th column of the matrix $\nv Q$ denotes the initial asset holdings of bank $i$. We then obtain the following first order approximation for the fire-sales externalities:
$D \approx MC^a-MC^f = Q_{\rm tot}^\top (\nv S Z)=(\nv Qv)^\top Z,$
where
\begin{equation}\label{eq:Definitionv}
v:={\rm Diag}(\kappa)\nv \alpha^\top \frac{Q_{\rm tot}}{\gamma\circ Q_0^{\rm nb}}\in\bbr^{N}.
\end{equation}
We refer to the vector $v$ as the \emph{systemic significance} of the $N$ banks, and discuss its properties and economic implications in Section~\ref{sec:systsign}. The vector $\nv{Q}v$, i.e., the initial allocation of assets within the banking sector weighted by their systemic significance\footnote{\Red Systemic significances can be interpreted as weights on the initial allocation of assets. However, they are typically not probability weights, i.e., generally they do not add up to 1.}, is a {\it network multiplier}: it describes how {\Red deviations from fundamentals due to} an initial price shock $Z$ {\Red are} amplified through the network of balance sheet holdings due to the leverage targeting actions of the banks. 

\noindent {\bf Deviation from Efficiency}: To quantify {ex-ante} these inefficiencies, we use the \emph{mean squared deviation} criterion $$\mathbb{E}[D^2]\approx\mathbb{E}[(MC^a-MC^f)^2]=:MSD(\nv Q).$$ We use the average of the squares of price deviation from fundamentals, i.e., we penalize equally positive and negative deviations of market prices from fundamental values.

We use $\mu$ and ${\bf \Sigma}$ to denote, respectively, the expected value and covariance matrix of shocks, i.e., $\mathbb{E}[Z]=\mu=(\mu_1,\ldots,\mu_K)^\top$, and ${\rm Cov}[Z]=\nv \Sigma\in\bbr^{K\times K}$. Henceforth, we impose the following assumption on the distribution of asset shocks.

\begin{assumption}\label{as:assets}
	Price shocks $Z_k$, $k=1,2,\dots, K$, are uncorrelated with variances $\sigma^2=(\sigma_1^2,\ldots,\sigma_K^2)^\top$, i.e., $\nv \Sigma= {\rm Diag}(\sigma^2)$, $\sigma_k>0$, and the initial asset prices are normalized to $P_0^k=1$,  $k=1,\ldots,K$.
\end{assumption}

\begin{remark}
{\Red Admittedly, Assumption~\ref{as:assets} is typically not satisfied for primary assets in a given market. However, one can construct uncorrelated auxiliary assets with normalized prices as portfolios of the original assets. We provide a sketch of this construction in E-Companion \ref{app:dependentshocks}. All results of this paper apply to these auxiliary assets with uncorrelated shocks and normalized prices, if both illiquidity characteristics and demand of the nonbanking sector are specified in terms of these uncorrelated assets.} 
\end{remark}

\subsection{Systemic Significance}\label{sec:systsign}
This section discusses the properties of the systemic significance of a bank, and its dependence on the model primitives.  
We start observing that, for any bank $i$, its systemic significance equals
	$
	v_i  =  \kappa^i \sum\nolimits_{k=1}^{K} \tfrac{\alpha^{ki}}{\gamma^k} \cdot \tfrac{Q_{\rm tot}^k}{ Q_0^{k,{\rm nb}}},
	$
as it easily follows from Eq.~\eqref{eq:Definitionv}. The systemic significance $v_i$ depends on the targeted leverage $\kappa^i$, banks' strategies $\alpha^{ki}$, the vector $\gamma$ of price elasticities, and the initial proportion of assets held by the nonbanking sector $\tfrac{Q_0^{k,{\rm nb}}}{Q_{\rm tot}^k}$. If bank $i$ liquidates a large fraction of an illiquid asset (i.e., $\tfrac{\alpha^{ki}}{\gamma^k}$ is high), then it creates a larger price pressure, especially if it is targeting a high leverage $\kappa^i$. If the nonbanking sector holds a significant fraction of the assets in the economy, then it will be able to better absorb the aggregate demand of the banking sector, and thus the systemic significance of any bank in the system will be reduced.

\begin{remark}
Observe that both $D$ and the square-root of the mean squared deviation $ \sqrt{\mathbb{E}[D^2]} \approx \sqrt{\mathbb{E}[(MC^a-MC^{f})^2]}$ are positively homogeneous, when viewed as functions of the systemic significance vector $v$. In particular, price pressures vanish and the market becomes efficient when $v$ approaches zero; conversely, as $v$ gets large, inefficiencies are higher.
\end{remark}


The mean and variance of market capitalization can be uniquely characterized by the matrix of asset holdings and the systemic significance of the banks in the system, as stated next.

\begin{lemma}\label{lem:MC}
	The mean and variance of price pressures are given by $\E[D] \approx (\nv Q v)^\top\mu$ and $\Var [D] \approx (\nv Q v)^\top\, \nv\Sigma \, (\nv Q v)$, respectively. This results in the following formula for the mean squared deviation 
	\begin{equation}\label{eq:MSE}	
		\mathbb{E}[D^2] \Red \approx  (\nv Q v)^\top \underbrace{(\mu\mu^\top+{\rm Diag}(\sigma^2))}_{=:\nv G\in\bbr^{K\times K}}(\nv Q v)=MSD(\nv Q).
	\end{equation}
\end{lemma}
\vspace{-0.7cm}
The formula above indicates that the mean squared deviation is a function of $\nv Q v$, i.e., the initial allocation of assets within the banking sector weighted by their systemic significance, and  of statistics about fundamental shocks collected in the matrix $\nv G$. 
This is consistent with intuition: larger shocks require banks to trade a higher amount of assets to restore their leverage targets. As a result, through the network multiplier $\nv Q v$, these shocks are amplified more and lead to a higher pressure on prices.

\section{f-Efficient Holdings}\label{sec:efficientholdings}
In this section, we study the impact of banks' portfolio holdings on market efficiency, and develop an explicit construction of holding matrices which minimize the deviation of market capitalization from its fundamental value. We refer to those matrices as {\it f-efficient holdings}.  
Section~\ref{sec:fhold} introduces a two-step procedure, ``aggregate first then allocate'' to {compute} f-efficient holdings. Section~\ref{sec:fulldivers} discusses the relation between f-efficient holdings and the diversification benchmark where each bank fully diversifies its portfolio holdings. 

\subsection{Characterization of f-Efficient Holdings}\label{sec:fhold}

Taking the initial budget of the banks as fixed, we compute the matrix $\nv Q$ of initial holdings which is f-efficient, i.e., which minimizes the mean squared deviation $MSD(\nv Q)$.\footnote{In E-Companion~\ref{app:liquidationstrategies}, we study the sensitivity of market inefficiencies to banks' trading strategies $\alpha^i$, $i=1,\ldots,N$, taking the initial holdings as given. Consistently with intuition, we find that it is f-efficient for a bank to sell solely its most liquid asset. 
}
 
Let $q\in\bbr^K$ be the vector of total initial holdings of the banking sector, and $b\in\bbr^N$ the vector of banks' initial budgets. The set of initially feasible asset allocations within the banking sector is then given by 
$\mathcal{D}=\mathcal{D}(q,b):=\{\nv Q\in\mathbb{R}^{K\times N}\mid \nv Q\nv{1}_N=q\quad\text{and}\quad \nv{1}_K^\top\nv Q=b^\top\},$
where we have assumed that the initial prices of all assets are  normalized (see Assumption~\ref{as:assets}). We then obtain  $\sum\nolimits_{k=1}^K q_k=\sum\nolimits_{i=1}^{N} b_i =: T.$
\begin{example}\label{ex:diversification_q}
Full portfolio diversification corresponds to the holding matrix
\begin{equation}\label{eq:Qdiversified}
\nv Q^{{\scriptscriptstyle \rm diversified}}:=\frac{1}{T} q b^\top=\begin{psmallmatrix}
\tfrac{b_1 q_1}{T} & \cdots & \tfrac{b_N q_1}{T}\\
\vdots & \ddots & \vdots\\
\tfrac{b_1 q_K}{T} & \cdots & \tfrac{b_N q_K}{T}
\end{psmallmatrix}\in\mathcal{D}.
\end{equation}
	In the special case that initial holdings are the same for all assets, i.e., $q_1=q_2=\ldots=q_K$, then full diversification is characterized by the holding matrix  $\nv Q^{{\scriptscriptstyle \rm diversified}}=\begin{psmallmatrix}
	\tfrac{b_1}{K} & \cdots & \tfrac{b_N}{K}\\
	\vdots & \ddots & \vdots\\
	\tfrac{b_1}{K} & \cdots & \tfrac{b_N}{K}
	\end{psmallmatrix}.$
\end{example}

 
\begin{definition} A feasible allocation matrix $\nv Q^*\in\mathcal{D}$ such that \begin{equation}\label{eq:fefficientholdings}
\nv Q^*\in\underset{\nv Q\in\mathcal{D}}{{\argmin}}\quad \underbrace{(\nv Q v)^\top \nv G (\nv Qv)}_{=MSD(\nv Q)}
\end{equation}
is called \emph{f-efficient}.
\end{definition}
\noindent f-efficient holdings are those which minimize, ex-ante, the mean squared deviation of asset prices from fundamentals among all feasible allocations. To exclude trivial cases and make the problem interesting, we make the following assumption (discussed in E-Companion \ref{app:ass:v}) which is typically satisfied in practice. 

\begin{assumption}\label{ass:v}
There exist some $i,j\in\{1,\ldots,N\}$ such that $v_i\not=v_j$, and it holds that $b^\top v\not=0$.
\end{assumption}

We next describe the two-step procedure used to construct f-efficient holding matrices:
\begin{itemize}
	\item {\bf Step 1: Aggregation.} In this step, we recover the network multiplier $y$ that minimizes the mean squared deviation, and which is consistent with the market structure. Concretely, let $\mathcal{D}_y:=\{y\in\bbr^K\mid \boldsymbol{1}_K^\top y=b^\top v\}.$ Such a multiplier is obtained from banks' initial holdings (and budgets), upon weighting them with the systemic significance vector, i.e.,
	\begin{align}\label{eq:yproblem}
	&y^*\in\underset{y\in\mathcal{D}_y}{\argmin}\quad y^\top \nv G y.
	\end{align}
	Because of this weighting, this multiplier accounts for the banks' leverage tracking behavior, their liquidation strategies, and illiquidity characteristics of the assets. We refer to the minimizing vector $y^*$ as the aggregate $v$-weighted holdings.

	\item {\bf Step 2: Allocation.} In this step, we identify an allocation of asset holdings to banks, which is consistent with the vector of aggregate $v$-weighted holdings obtained from the previous step. Specifically, we denote by $\nv Q^*\in\mathcal{D}$ the matrix, consistent with the budget and supply constraints, which distributes 
	the aggregate $v$-weighted holdings to individual banks  according to their systemic significance, i.e., $\nv Q^* v=y^*$. 
\end{itemize}
The decomposition discussed above presents both conceptual and computational advantages. From a conceptual perspective, observe that the matrix $\nv G$ of shock statistics and the vector $v$ of banks' systemic significance are ``sufficient statistics'' for the problem. In the aggregation step, the minimized functional only depends on $\nv G$, while $v$ only enters into the constraint set. The allocation step instead, takes the aggregate holdings computed from the previous step as given, and determines the holding matrix only on the basis of $v$. From a computational point of view, observe that finding the f-efficient holdings requires solving a $K \times N$ dimensional quadratic problem with linear constraints. Using the proposed decomposition, we first solve a $K$ dimensional quadratic problem with one linear constraint, and subsequently solve a simple $K \times N$ dimensional unconstrained linear system.

The following proposition states that the two-step procedure described above identifies f-efficient holdings. 
\begin{proposition}\label{prop:solutionmethod}
	Let $N\ge 2,K\ge 1$. The following statements are equivalent:
	\begin{enumerate}
		\item[(i)] $\nv Q^*\in\mathcal{D}$ is f-efficient.
		\item[(ii)] $y^*=\nv Q^*v$ for some $\nv Q^*\in\mathcal{D}$, and $y^*$ solves the  problem \eqref{eq:yproblem}.
	\end{enumerate}
\end{proposition}

In the rest of the section, we discuss in more detail each step of the procedure above, and highlight the key economic insights. 

\subsubsection{Step 1: Aggregation}
We first show that the minimization  problem \eqref{eq:yproblem} admits a unique solution. 
\begin{lemma}\label{lem:aggregateholdings}
	The unique solution to problem \eqref{eq:yproblem} is given by $y^*:=\frac{b^\top v}{\nv{1}_K^\top z}\cdot z,$ with $z:=\nv G^{-1}\nv{1}_K\in\bbr^K.$
\end{lemma}
The above expression indicates that aggregate $v$-weighted holdings are determined by the vector $v$, capturing the systemic significance of banks, weighted by the budget each bank is endowed with, and by the inverse of the matrix $\nv G$ which captures the size of the exogenous price shocks. Specifically, the aggregate systemically weighted budget vector $b^\top v$ is split into the $K$ available assets through the inverse of $\nv G$.

\begin{example}
	Consider the special case of zero mean shocks, i.e., set $\mu=0$. Then $y^*_k=(b^\top v) \frac{1/\sigma^2_k}{\sum\nolimits_{\ell=1}^K 1/\sigma_\ell^2}$. Hence, the higher the variance of price shocks $\sigma_k^2$, the lower the fraction of asset $k$ in the aggregate $v$-weighted holdings portfolio $y^*$. This is intuitive: an asset creating high price pressure when banks manage their assets to restore their target leverage should, in aggregate, be invested less. 
\end{example}

\subsubsection{Step 2: Allocation}

In the second step, the aggregate $v$-weighted holdings $y^*\in\bbr^K$ are allocated to the individual banks.

\begin{proposition}\label{prop:existenceallocation}
	For every $N\ge 2, K\ge 1$, there exists a matrix $\nv Q^*\in\mathcal{D}$ satisfying $\nv Q^* v=y^*$.
\end{proposition}

In the proof of the proposition, we construct a particular solution $\nv Q^*$, and describe the structure of the linear subspace of solutions. The following theorem quantifies the mean squared deviation achieved by a matrix of f-efficient holdings, and addresses the uniqueness of the allocation.
\begin{theorem}\label{thm:efficientholdings}
	Let $N,K \ge 2$.
	\begin{itemize}
		\item[a)] There exists an f-efficient holding matrix $\nv Q^*$ with mean squared deviation
		$MSD(\nv Q^*)=\tfrac{(b^\top v)^2}{\nv{1}_K^\top\nv G^{-1}\nv{1}_K}.$
		\item[b)] The f-efficient holding matrix is unique, if and only if there are exactly $N=2$ banks. In this case, the unique f-efficient holding matrix is given by
		$\nv Q^{N=2}:=\tfrac{1}{v_2-v_1}\begin{psmallmatrix}
		v_2 q-y^* & y^*-v_1 q\\
		\end{psmallmatrix}\in\bbr^{K\times 2}.
		$
	\end{itemize}
\end{theorem}

\subsection{When is Diversification f-Efficient?}\label{sec:fulldivers}

In this subsection, we provide the conditions under which full diversification is f-efficient. 
\begin{theorem}\label{thm:diversification}
Full diversification $\nv Q^{\rm diversified}$ is f-efficient, if and only if $q$ and $z=\nv G^{-1}\nv{1}_K$ are linearly dependent.
\end{theorem}
A direct consequence of the above theorem is that full diversification is f-efficient if the system is completely homogeneous, i.e., all asset shocks have the same mean and variance, and the total initial holdings of the banks are the same. Full diversification is no longer f-efficient if a little amount of heterogeneity is introduced in the system. 
\begin{corollary}\label{cor:diversification}
	If either
	\begin{itemize}
		\item[a)] $q_1=\ldots=q_K$, $\sigma_1^2=\ldots=\sigma_K^2,$ $\mu_1=\ldots=\mu_{K-1}$ and $\mu_K=\mu_1+\varepsilon$ with $\varepsilon\not=-K \mu_1$, or
		\item[b)] $\mu_1=\ldots=\mu_{K}$, $\sigma_1^2=\ldots=\sigma_K^2,$ $q_1=\ldots=q_{K-1}$ and $q_K=q_1+\varepsilon$, or
		\item[c)] $q_1=\ldots=q_K$, $\mu_1=\ldots=\mu_{K}$, $\sigma_1^2=\ldots=\sigma_{K-1}^2$ and $\sigma_K^2=\sigma_1^2+\varepsilon$,
	\end{itemize}
	then $\nv Q^{\rm diversified}$ is f-efficient, if and only if $\varepsilon=0$.
\end{corollary}

The result in the above corollary is consistent with intuition. If assets are fully homogeneous and the total holdings of the banking sector in each asset are the same, there is no reason to prefer one asset over the other. In this case, full diversification minimizes the portfolio liquidation risk and is optimal. However, as soon as assets no longer have identical characteristics, an f-efficient allocation requires to allocate assets to banks in accordance with their systemic significance. 

\section{Comparative Statics and Examples}\label{sec:casestudies}

In this section, we construct case studies to analyze the structure of f-efficient holdings. In the first part of this section, we consider the stylized case $N=K=2$, and study the distance of f-efficient holdings from diversification as a function of key model parameters. For any matrix $\nv Q\in\bbr^{K\times N}$, we measure the distance from full diversification by the Frobenius norm
$d(\nv Q):=\|\nv Q-\nv Q^{\rm diversified}\|_F.$ In the second part of the section, we summarize our findings on the minimal distance of f-efficient holdings from full diversification for the case $N=K=3$. In the third part, we compare the impact of f-efficiency and diversification on the distribution of market capitalization under different economic scenarios. 

\subsection{The Case $N=K=2$}\label{subsec:NK2}

According to Theorem \ref{thm:efficientholdings} b), the f-efficient holding matrix for $N=K=2$ is unique and given by
$
\nv Q^{2\times2}:=\tfrac{1}{v_2-v_1}\begin{psmallmatrix}
v_2 q_1-y_1^* & y_1^*-v_1 q_1\\
v_2 q_2-y_2^* & y_2^*-v_1 q_2\\
\end{psmallmatrix}.
$
We normalize the total supply of assets within the banking sector, and the budgets of banks to $q_1=q_2=b_1=b_2=:x>0$;  following \cite{capponi2015price}, we choose $x=0.08$, where the total supply is normalized to 1 for each asset. Hence, the 
the total size of the banking sector is 8\% of the total size of the system. 
  
\subsubsection{Systemic Significance and Asset Riskiness}\label{subsec:Impactsigma}
We start with an exploratory analysis, where we plot the f-efficient holdings of bank 1 and the distance of f-efficient holdings from diversification $d(\nv Q^{2\times 2})$ as a function of the riskiness of the first asset (Figure \ref{fig:allocationSigma}).
\begin{figure}[htb]
	\centering
	\includegraphics[width=0.5\linewidth]{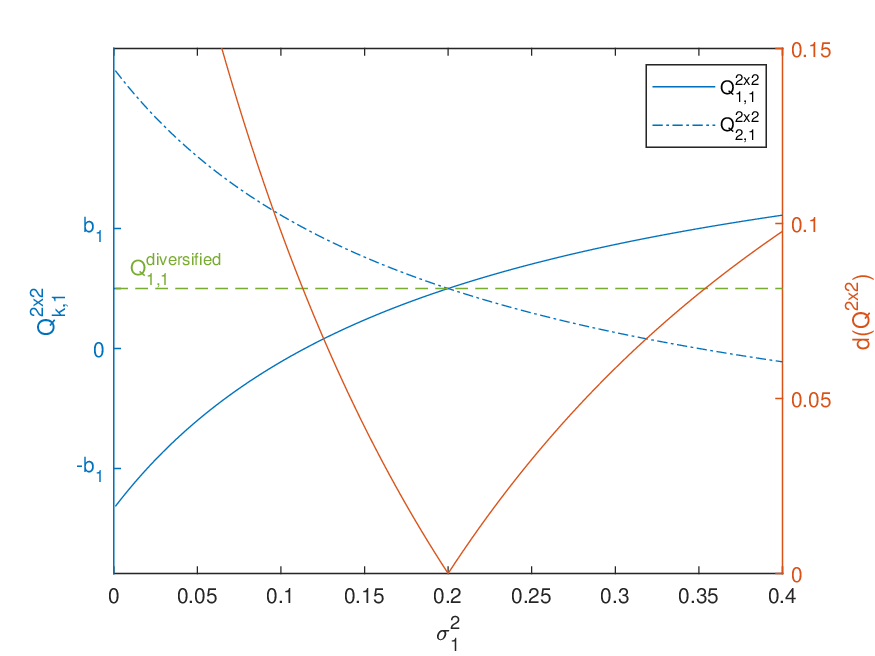}
	\caption{f-efficient holdings of bank 1 ${\nv Q}_{k,1}^{2\times 2}$ for assets $k=1,2$ and distance from diversification, $d(\nv Q^{2\times2})=\|\nv Q^{2\times2}-\nv Q^{\rm diversified}\|_F$, as a function of $\sigma_1^2$. We fix $\sigma_2^2=0.2$, $\mu=(0,0)^\top,$ $v=(0.04,0.07)^\top$, $q=b=(0.08,0.08)^\top$.}
	\label{fig:allocationSigma}
\end{figure}
Consistent with Corollary \ref{cor:diversification}, 
full diversification is f-efficient if and only if $\sigma_1^2=\sigma_2^2=0.2$ (cf. Corollary \ref{cor:diversification}).  Figure \ref{fig:allocationSigma} additionally suggests that a systemically more significant bank would have lower  f-efficient holdings in the riskier asset than a systemically less significant bank: as $\sigma_1$ increases, bank $2$ decreases its holdings in asset 1, while bank $1$ increases its holdings in that asset. We rigorously formalize these observations in the following lemma. 

\begin{lemma}\label{lem:sigma2}
	Let $q_1=q_2=b_1=b_2=x>0$ and $\mu_1=\mu_2$.
	\begin{itemize}
		\item[a)] For $\sigma_1>0$, we have $\tfrac{\partial }{ \partial \sigma_1} d(\nv Q^{2\times 2})\lessgtr 0,\text{ if $\sigma_1^2 \lessgtr \sigma_2^2$}.$ 
		\item[b)] If $v_2>v_1>0$, then $\frac{\partial  \nv Q^{2\times 2}_{11}}{\partial \sigma_1}>0$ for $\sigma_1>0$.
	\end{itemize}
\end{lemma}

\subsubsection{Systemic Significance and Shock Sizes}\label{subsec:Impactmu}
As in the previous subsection, we start with a graphical illustration of the sensitivity of bank 1's f-efficient holdings and their distance from diversification to changes in the expected shock size $\mu$. 
\begin{figure}[htb]
	\centering
	\includegraphics[width=0.5\linewidth]{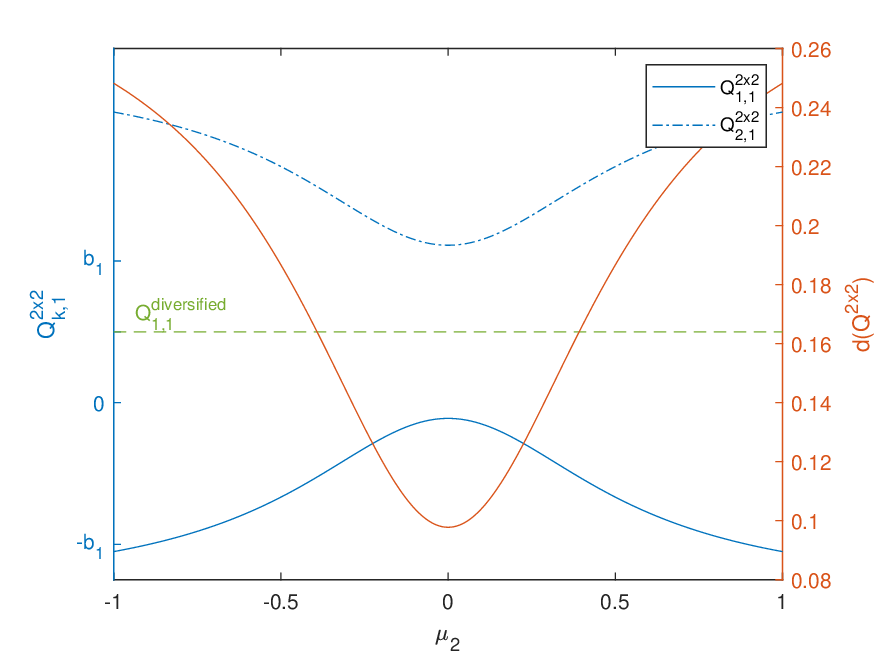}
	\caption{f-efficient holdings of bank 1 ${\nv Q}_{k,1}^{2\times 2}$ for assets $k=1,2$ and distance from diversification $d(\nv Q^{2\times2})=\|\nv Q^{2\times2}-\nv Q^{\rm diversified}\|_F$ as a function of $\mu_2$ for fixed $\sigma^2=(0.1,0.2)^\top$, $\mu_1=0,$ $v=(0.04,0.07)^\top$, $q=b=(0.08,0.08)^\top$.}
	\label{fig:allocationMu}
\end{figure}
Figure \ref{fig:allocationMu} indicates that, as the expected (absolute) size of price shocks for asset 2 increases, the f-efficient holdings of the least systemically significant bank (i.e., bank 1) increase. This can be, again, understood in terms of the banks' systemic significance: a more systemically significant bank, i.e., one that tracks a higher leverage ratio or which trades a larger fraction of illiquid assets, should reduce its holdings of asset 2, because the trading actions in response to the shock impose a higher pressure on the price, and hence large deviation of prices from fundamentals.

The distance from full diversification is minimal when the system achieves the highest possible degree of homogeneity, i.e., $\mu_2=\mu_1$. As the system becomes more heterogeneous, the distance increases. As shown in Corollary \ref{cor:diversification}, the minimal distance converges to zero as the system becomes fully homogeneous, i.e., $\sigma_1=\sigma_2$. We formalize these observations via the following lemma.
\begin{lemma}\label{lem:mu}
	Let $q_1=q_2=b_1=b_2=x>0$, $\sigma_1^2=\sigma_2^2$ and $\mu_1=0$.
	\begin{itemize}
			\item[a)] It holds that $\tfrac{\partial}{\partial\mu_2} d(\nv Q^{2\times 2}) \lesseqqgtr 0$, if $\mu_2 \lesseqqgtr 0$.
	\item[b)] If $v_2>v_1>0$, then $\tfrac{\partial \nv Q^{2\times 2}_{11}}{\partial\mu_2} \gtreqqless 0$, if $\mu_2 \lesseqqgtr 0$.
\end{itemize}
		
\end{lemma}


\subsubsection{The Influence of Systemic Significance on f-Efficient Holdings}\label{subsec:Impactv}
We analyze how heterogeneity in systemic significances impacts the degree of diversification of banks' holdings.
\begin{figure}[htb]
	\centering
	\includegraphics[width=0.5\linewidth]{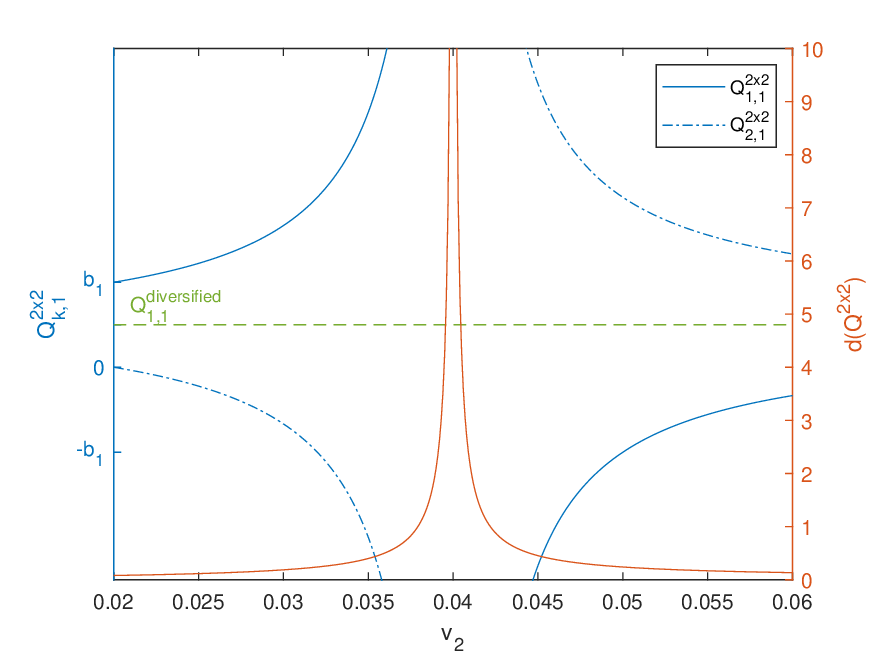}
	\caption{f-efficient holdings of bank 1 ${\nv Q}_{k,1}^{2\times 2}$ for assets $k=1,2$ and distance from diversification $d(\nv Q^{2\times2})=\|\nv Q^{2\times2}-\nv Q^{\rm diversified}\|_F$ as a function of systemic significance $v_2$ for fixed $\sigma^2=(0.1,0.2)^\top$, $\mu=(0,0)^\top,$ $v_1=0.04$, $q=b=(0.08,0.08)^\top$.}
	\label{fig:allocationV}
\end{figure}
Figure~\ref{fig:allocationV} highlights that, as the two banks become closer in terms of systemic significance, the distance of f-efficient holdings from diversification increases. This highlights the fundamental role of systemic significance on banks' f-efficient holdings: if two banks are similar in terms of systemic significance (for instance because they adopt similar trading strategies), then it is beneficial to sacrifice diversification benefits to reduce portfolio overlapping, and thus price pressures. These intuitions can be formalized via the following lemma. 

\begin{lemma}\label{lem:v}
	Let $q_1=q_2=b_1=b_2=x>0$, $v_1>0$ and $|z_1|\not=|z_2|$\footnote{A sufficient condition for $|z_1|\not= |z_2|$ is that $\mu_1=\mu_2$ and $\sigma_1^2\not=\sigma_2^2$.}.
	For $0<v_2\not= v_1$, it holds that $\tfrac{\partial}{\partial v_2} d(\nv Q^{2\times 2}) \gtrless 0$, if $v_2 \lessgtr v_1$.
\end{lemma}

\subsection{The Case $N=K=3$}\label{subsec:NK3}
Having established the results for a system consisting of two banks and two assets, we analyze numerically how the findings would change for a larger economy. If the number of banks is $N>2$, f-efficient holdings are no longer unique (see Theorem \ref{thm:efficientholdings} b)). In this case, we consider the f-efficient holdings whose Frobenius distance from diversification is minimal. 

Noticeably, the qualitative findings in the case $N=K=3$, shown in E-Companion \ref{app:NK3}, remain similar to the setting $N=K=2$. The f-efficient holdings get farther away from a full diversification strategy if heterogeneity in banks' systemic significance decreases. The intuition behind the result remains unchanged, i.e., in a system where banks are systemically very close, a full diversification strategy for each bank may lead to larger price pressures because all banks rebalance their portfolios in a similar fashion to meet their leverage targets. 

In E-Companion~\ref{app:nonuniqueness}, we analyze f-efficient holdings in an economy with more than two banks. We find that any f-efficient allocation prescribes the most systemically significant bank to have higher holdings of the less risky asset, and the least systemically significant bank to have higher holdings of the more risky asset. We also show that for a system of three banks, given a fixed f-efficient allocation, any other f-efficient allocation is obtained by shifting the holdings within each bank based on the difference in systemic significances of the other two banks. 
 
{\Red We conclude by discussing the implications of non-uniqueness. The set of f-efficient holdings includes allocation profiles which are (i) closer to full bank diversity, (ii) closer to full diversification of individual banks' portfolios, or (iii) closer to the observed banks' portfolios. Such multiplicity of f-efficient holdings gives a regulator flexibility to achieve the desired outcome of low market inefficiency. For instance, the regulator may find it cheaper and quicker to align banks' incentives with the socially desirable outcome, if he steers the system towards the matrix of f-efficient holdings which is closer to that currently observed in the banking system. Such a holding matrix may be estimated from the forms FR Y-9C, which financial institutions file with the Federal Reserve every quarter. These forms provide consolidated information on each bank's exposures, and are available through the Board's Freedom of Information Office.\footnote{\Red \cite{duarte2018fire} build their empirical model on FR Y-9C {balance sheet} data.} 
}

\subsection{Systemic Significance, Market Scenarios and Asset Holdings} \label{sec:assetholdings}

In this section, we analyze how the distribution of banks' holdings depends on banks' systemic significance, heterogeneity in the distributions of initial shocks, and illiquidity characteristics of the assets. We also validate the accuracy of the first order approximation of market capitalization used throughout the paper. 
Our results indicate that the mean squared deviation of the {\it actual} market capitalization from its fundamental value is low if the matrix of bank holdings coincides with the f-efficient holdings. This indicates that the solution to the (approximate) optimization problem, i.e., where the first order approximation of the systemicness matrix is used, yields a low value for the actual objective function where the exact expression of the systemicness matrix is used.

In the analysis below, we consider three asset classes, each consisting of assets with identical characteristics, for a total of ten assets. As a result, we demonstrate that the methodology proposed in this paper scales well to economies larger than those considered (analytically) in earlier sections.

\subsubsection{Market Setting}\label{subsec:marketsetting}
We consider a financial market consisting of $N=2$ banks.  
We choose $K=10$, normalize the total supply of each asset to 1, and set the holdings of banks in each asset to $0.08$. All parameters are consistent with \cite{capponi2015price}. The two banks are assumed to have the same budget ($b_1=b_2=0.8$), and the vector of leverage targets is $\kappa=(9,10)^\top$.

Banks are assumed to follow a proportional liquidation strategy, i.e., $\alpha^{ki}=1/10$ for all assets $k=1,\ldots,10$, and both banks $i=1,2$. The assets belong to three different groups: Assets 1 and 2 belong to group 1, assets 3 through 8 to group 2, and assets 9 and 10 to group 3. Within each of the three asset groups, the illiquidity characteristics of the assets, and the expectation and variance of asset price shocks are equal. 

We consider three scenarios, each characterized by a certain value of the shock variance and liquidity of the assets. We refer to the three scenarios as liquidity, intermediate crisis, and high risk high illiquidity scenarios. Across all scenarios, we set $\mu=(0.1,0.1,0.2,0.2,0.2,0.2,0.2,0.2,0.3,0.3)^\top$. 

The numerical values of the asset illiquidity characteristics should be thought as normalized to the corresponding characteristics of a reference asset, and are broadly consistent with the estimates reported in Table 4 of \cite{duarte2018fire}.\footnote{Their estimates are based on the Net Stable Funding Ratio of the Basel III regulatory framework. Their illiquidity parameter is the reciprocal of ours, i.e., in their setting a larger value  corresponds to a higher illiquidity of the asset. They take US Treasuries as the reference asset, i.e., the price impact of U.S. Treasuries is normalized to 1.} Consistent with empirical evidence, the higher the asset illiquidity (i.e., the lower $\gamma$), the larger the variance of the exogenous asset price shocks, capturing the fact that more illiquid securities have a higher volatility than liquid ones. 

\begin{itemize}
	\item[(L)] \textbf{Liquidity Scenario:}\\
	Banks invest in liquid assets, i.e, assets with high price elasticities, or equivalently, low illiquidity characteristics. The first asset class has the highest price elasticity and the lowest shock variance, and the third asset class has the lowest price elasticity and the highest variance. The second asset class has an intermediate value for those two quantities. Specifically, we choose
	$\gamma^L=(9,9,8,8,8,8,8,8,7,7)^\top,$ $\sigma^L=(0.1,0.1,0.2,0.2,0.2,0.2,0.2,0.2,0.3,0.3)^\top.$

	\item[(I)] \textbf{Intermediate Crisis Scenario:}
	In contrast to scenario (L), the third asset class is significantly more illiquid (i.e., lower price elasticity) and has higher shock variance. This situation is typical of the beginning of a crisis period, where one asset may experience a severe shock and then become hard to sell quickly due to the lack of outside investors (the nonbanking sector in our model) willing to purchase the asset.\footnote{For instance, the volume of agency mortgage backed securities, typically highly liquid assets, declined substantially from 2008 to 2014, which is an indicator of worsening liquidity.} The parameters corresponding to the other asset classes are not altered. This scenario is specified by
	$\gamma^{L}=(9,9,8,8,8,8,8,8,1,1)^\top,$ $\sigma^{L}=(0.1,0.1,0.2,0.2,0.2,0.2,0.2,0.2,1,1)^\top.$

	\item[(H)] \textbf{High Risk High Illiquidity Scenario:}
	Banks invest in assets with high illiquidity and volatility. This captures, for instance, a situation where banks invest in securities such as non-agency based mortgage, municipal bonds, or commercial and industry loans. All asset classes are thus characterized by a lower price elasticity, $\gamma^{H}=(3,3,2,2,2,2,2,2,1,1)^\top,$ and higher shock variance, $\sigma^{H}=(0.9,0.9,1.1,1.1,1.1,1.1,1.1,1.1,1.2,1.2)^\top$, compared to the previous two scenarios.
	 
\end{itemize}
In E-Companion \ref{app:fefficientmarketholdings}, we explicitly compute banks' systemic significances and f-efficient holdings for each of the above defined scenarios. Overall, we find that the higher the systemic significance of a bank, the lower its holdings of the safer asset relative to the riskier asset. In addition, in scenario (H), there is little heterogeneity in the riskiness of the assets, and high heterogeneity in banks' systemic significance. As a result, the f-efficient holdings in this scenario are more evenly distributed, i.e., closer to full diversification (see also the values of the distances given in Table \ref{tab:statistics} below).

\subsubsection{Diversification and  f-Efficiency}\label{sec:approxanalysis}
In this section, we compute the {\it exact} market capitalization, i.e., $MC^e=Q_{\rm tot}^\top(P_0+(\nv I-\nv S)^{-1}Z),$ both under f-efficient $(\nv{Q}^{*,\cdot})$ and fully diversified ($\nv Q^{\rm diversified}$) holding matrices. We suppose that the vector of shocks $Z$ follows a multivariate normal distribution with mean vector $\mu$, which is the same across all scenarios, and covariance matrix ${\rm Diag}(\sigma^2)$, where $\sigma$ depends on the considered scenario. In each scenario, we draw $100,000$ independent samples of the shock vector $Z$. 

We report the relevant statistics in Table \ref{tab:statistics}. Figure \ref{fig:AllScenarios} also provides a comparison of the probability density functions of market inefficiency and the corresponding box plots across all three scenarios. Consistent with intuition, the variance of market capitalization is the smallest in the liquidity scenario, and the highest in the high risk high illiquidity scenario. We find that diversification results in a higher variance than f-efficiency across all market scenarios. While this difference is not very significant in the liquidity scenario, it becomes considerable in the high risk high illiquidity scenario, especially if (as in the intermediate crisis scenario) the assets are shocked heterogeneously.

Observe, from the first row of the table, that the distance between f-efficient and fully diversified holding matrices is the lowest in the scenario (H). Nevertheless, the mean-squared deviation criterion and variance of the exact market capitalization is the highest (in absolute terms) in such a scenario. Taken together, these two observations imply that if banks invest in high risk high illiquid assets, it suffices to only slightly distort the holdings from f-efficiency to induce large increases in the variance and mean squared deviation of market capitalization from its fundamental value. This is because holdings that are not f-efficient, such as fully diversified holdings, create a price pressure which becomes increasingly larger as we move towards  balance sheets with highly volatile and illiquid assets.

\begin{table}[h!] 
	\begin{center}
		\centering
		\begin{tabular}{llll|ll}\toprule
			& L & I & H & I/L & H/L\\\midrule
			$d(\nv Q^{*,\cdot})$ & 8.61   & 8.74 &    1.08 &   1.02 & 0.13\\
			$\mathbb{E}(MC^{e})$ for $\nv Q^{*,\cdot}$ & 	12.07 &  12.20  & 13.45 & 1.01 & 1.11\\
			$\mathbb{E}(MC^{e})$ for $\nv Q^{\rm diversified}$ &  12.23 &  12.65 &  13.74 & 1.03 & 1.12 \\
			${\rm Var}(MC^{e})$ for $\nv Q^{*,\cdot}$ & 	0.44 &    2.23 &  37.20 & 5.07 & 84.55\\
			${\rm Var}(MC^{e})$ for $\nv Q^{\rm diversified}$ & 0.55  &  3.95 &  41.39 & 7.18 & 75.25\\
			$\mathbb{E}[(MC^{e}-MC^{f})^2]$ for $\nv Q^{*,\cdot}$ & 	0.02  &  0.10  &  9.66 & 5.00 & 483.00\\
			$\mathbb{E}[(MC^{e}-MC^{f})^2]$ for $\nv Q^{\rm diversified}$ & 	0.06  &  0.66 &  12.14 & 11.00 & 202.33\\\bottomrule
		\end{tabular}	
		\caption{Statistics from the Monte Carlo simulation used to estimate the exact market capitalization. The Frobenius distance $d(\nv Q^{*,\cdot})$ measures the difference between diversified and f-efficient holdings. The last two columns present the numbers relative to the liquidity scenario.}	
		\label{tab:statistics}
	\end{center}
\end{table}

\begin{figure}
	\vspace{-0.1cm}
	\centering
	\subfloat[]{\includegraphics[width=0.5\textwidth]{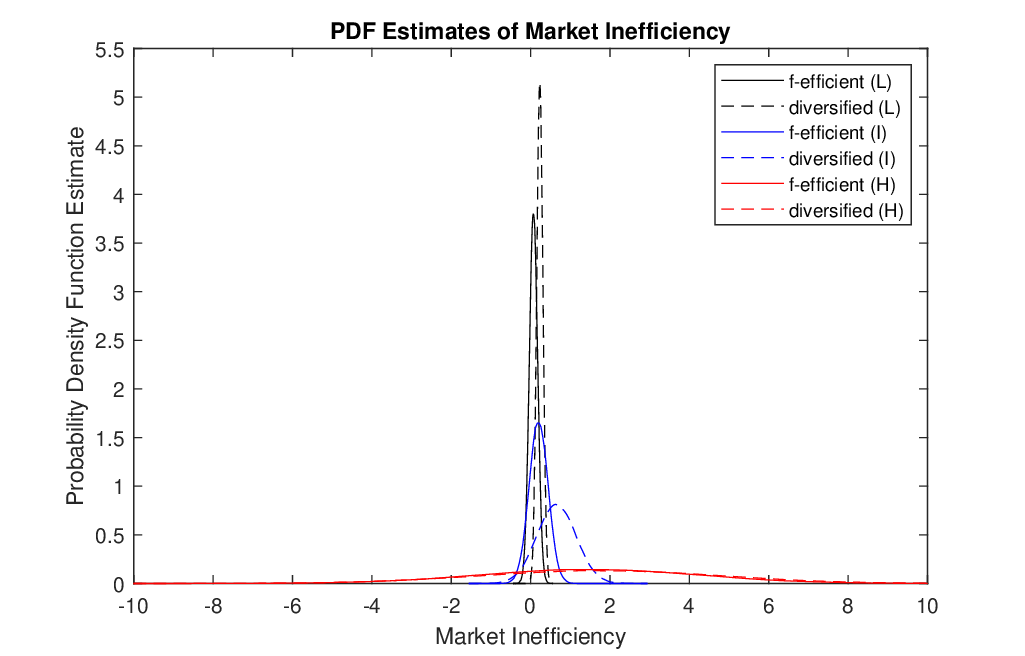}}
	\subfloat[]{\includegraphics[width=0.5\textwidth]{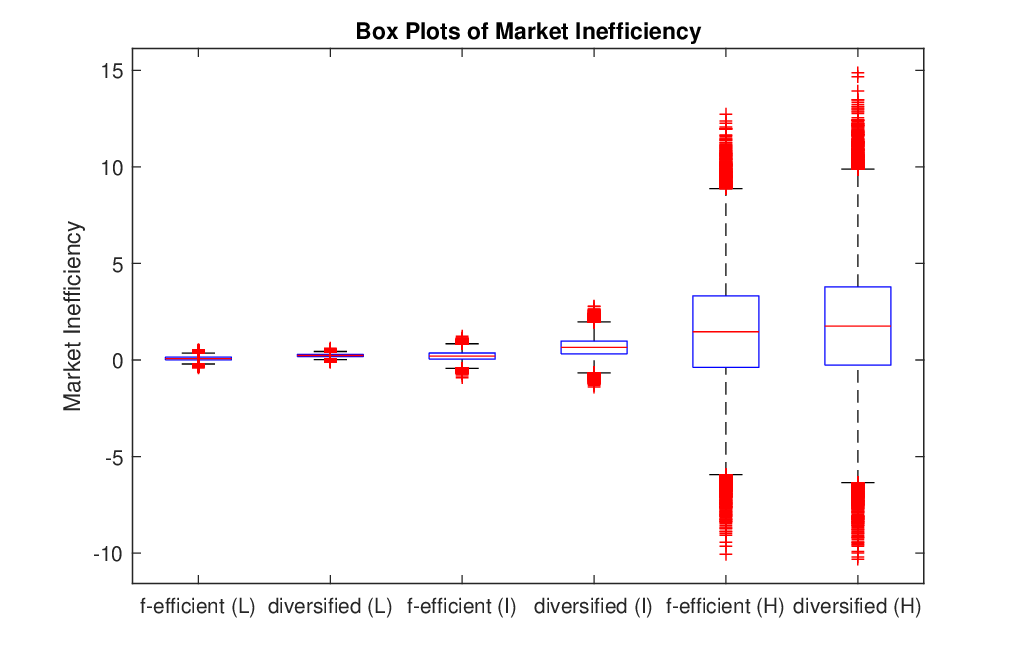}}
		\vspace{-0.1cm}
	\caption{Probability density function estimates (a) and box plots (b) of market inefficiency $MC^e-MC^f$ for f-efficient $(\nv{Q}^{*,\cdot})$ and fully diversified ($\nv Q^{\rm diversified}$) holdings, generated from 100,000 samples of $Z\sim\mathcal{N}_{10}(\mu,{\rm Diag}((\sigma^{\cdot})^2))$, in the three considered scenarios: (L) liquidity, (I) intermediate crisis, and (H) high risk high illiquidity.}
	\label{fig:AllScenarios}
\end{figure}

\section{Empirical Insights and Policy Implications}\label{sec:empirical}

We explain how the construction of our metrics is supported by existing empirical literature and highlight policy implications of our model. 

\noindent\textbf{Systemic significance.} The systemic significance of a bank increases with (i) the targeted leverage, (ii) the illiquidity characteristics of assets held, and (iii) the weights of its trading strategy on illiquid assets.  \cite{duarte2018fire} calibrate a measure of bank systemicness, which is high if the bank has a
high asset value and leverage, holds a large amount of illiquid assets, and is subject to a large initial shock. Their measure and ours share similarities, as they are both based on the same bank-specific characteristics (e.g., leverage and size), and bank-asset characteristics (illiquid portfolio overlapping). However, \cite{duarte2018fire} consider an ex-post scenario where a negative shock has already hit the system, and forced banks to deleverage. By contrast, we develop an ex-ante analysis where the distribution of future shocks -- but not their realization -- is known. Nevertheless, the empirical analysis of \cite{duarte2018fire} supports our claims that these three key dimensions should play an important role in measuring systemic risk. They find that the excessive growth of the banking sector is one of the main determinants for the increase of aggregate vulnerability from fire-sale spillovers. Banks' leverage and illiquidity concentration increased in the years leading to the crisis and declined again after 2008. They also find that Citigroup and Bank of America have the highest systemicness. Both of these banks increased leverage pre-crisis, and reduced both leverage and concentration on illiquid assets after the crisis. Follow-up empirical studies can calibrate the systemic significance measure identified by our study, and compare it against the measure of \cite{duarte2018fire} both in the cross section and in the time series (i.e., for different historical episodes). 

\noindent\textbf{Diversity.}  Our analysis indicates that a banking system in which each bank diversifies its  portfolio reduces asset price distortion only if there is enough systemic heterogeneity across banks. While banks are likely to target different leverages, they may adopt similar trading strategies when they adjust their balance sheets to restore their leverage targets. For instance, they may follow a pecking order of liquidation during fire-sale periods, where each bank tries to sell the most liquid assets first to minimizing the price impact of fire sales. Differently, if banks engage in asset fire purchases in response to a positive shock, they might initially purchase assets with high positive price pressure to increase returns. Alternative liquidation rules are also possible. For instance, regulatory requirements on risk-weighted assets such as liquidity coverage ratios provide incentives to all banks to sell assets with high risk-weights first (see, for instance, \cite{cifuentes2005liquidity}), and these assets are typically more illiquid. Regardless of the specific liquidation strategy, it is highly likely that banks adopt very similar trading strategies, which in turn reduces the level of systemic heterogeneity in the system. Under these circumstances, portfolio diversification within each bank becomes less desirable than a diverse system with little overlapping on assets with high illiquidity characteristics. 
This indicates that the disclosure of detailed information on changes in banks' portfolios would be beneficial for asset price stability. A regulator may then design policies which incentivise banks to hold portfolios where diversification benefits are sacrificed to achieve a more diverse system, even more if avalable information indicates that banks adopt very similar trading strategies.

\noindent\textbf{Adaptive capital buffers.} Large shocks and systemic significances imply high deviations of prices from fundamentals. Our model demonstrates that shocks hitting one bank spill over to other banks through the network of asset prices and harm price stability. The more systemically significant a bank is, the larger these spillovers are likely to be. Our analysis thus indicates that the government should mandate adaptive capital requirements. More systemically significant banks should hold larger capital buffers to avoid excessive price distortion. When shocks occur, governments can ease capital buffers to allow banks providing additional liquidity and intermediation services.  Another policy commitment device which can be used to stabilize prices in asset markets are large-scale asset purchase programs. Those were used extensively during the 2008 financial crisis especially in regards to Treasury securities. Based on our findings, the price stabilizing effect of these asset purchases would be stronger if they are executed when large shocks are expected to hit systemically significant banks.

{\Red
\noindent\textbf{f-efficiency and individually optimal holdings.} A social planner would optimally choose f-efficient holdings to minimize market inefficiency. If one were to account for banks' strategic behavior, the resulting equilibrium profile of banks' holdings may be far from f-efficient. Over longer time horizons, each bank would optimally choose its holdings to maximize a desired objective while accounting for regulatory constraints. For example, each bank may maximize its return on equity over a certain time period, subject to financial constraints (e.g., \cite{Froot}) or risk constraints.\footnote{\Red \cite{adrian2010liquidity} show that procyclical leverage follows directly from the counter-cyclical nature of unit value at risk.}

Denote by $\nv Q^{\rm cur}$ the holdings at the beginning of the time horizon that is relevant for regulatory purposes, and by
$\nv Q^{{\rm end}}$ the holdings at the end of this period when banks have adjusted their portfolios based on their objectives. The MSD criterion defined in \eqref{eq:MSE} is used to measure the influence of banks' choices on the deviations from fundamentals. Recall that the MSD depends on banks' holdings, the first two moments of exogenous price shocks, and the vector of systemic significances. In the following discussion, we assume that those quantities do not change significantly over the regulation period.

Set $\Delta Q^i=Q^{{\rm end},i}-Q^{{\rm cur},i}$. Using a first order Taylor series expansion, the impact of bank $i$ on the MSD can be approximated by the size $h= \|\Delta Q^i\|$ of the change in its holdings times the derivative of the MSD at $\nv Q^{\rm cur}$ in the direction of the unit vector $r = \tfrac{\Delta Q^i}{\|\Delta Q^i\|}  \in\bbr^K$, i.e., 
$D_r^iMSD(\nv Q^{\rm cur})=\left(\frac{\partial MSD}{\partial Q^i}(\nv Q^{{\rm cur}}) \right)^\top r,
$
where $\frac{\partial MSD}{\partial Q^i}$ denotes the gradient of the MSD with respect to bank $i$'s holdings $Q^i$. An explicit calculation yields
$D_r^iMSD(\nv Q^{\rm cur})=2\cdot v_i\cdot \sum_{j=1}^{N}v_j\cdot  ( Q^{{\rm cur},j})^\top {\nv G} r.$
The directional derivative depends on systemic significances, current holdings of all banks, and on the first two moments of asset shocks. The higher the systemic significance of a bank, the higher the directional derivative and, thus, the stronger the impact of its portfolio choice on the MSD. Directional derivatives capture how individual banks' actions influence the MSD and can, thus, be used as a basis for regulatory schemes. We next discuss two approaches to steer banks' holdings in the direction of f-efficiency: directional taxes and directional certificates.

{\bf Directional taxes.} For some multiplier $c_{\rm Tax}>1$, we define
$T_i(r,h)=c_{\rm Tax}\cdot D_{r}^i MSD(\nv Q^{\rm cur})\cdot h=V^i_{\rm Tax} r\cdot h$, where the vector $V^i_{\rm Tax}=c_{\rm Tax}\cdot (2\cdot v_i\cdot \sum_{j=1}^{N}v_j\cdot  ( Q^{{\rm cur},j})^\top {\nv G})$
determines how subsidies (negative sign) and taxes (positive sign) are computed from the direction $r$ and the size $h$ of the changes in the portfolio of bank $i$. The subsidy or tax $T_i(r,h)$ is proportional to the size $h$ of the changes of bank $i$'s portfolio and the derivative of the MSD in the direction $r$ of these changes. The scheme provides incentives for banks to move towards f-efficiency. Variations of this tax scheme include convex tax and subsidy schemes (e.g., progressive taxes), where banks are penalized or rewarded increasingly more, the more their actions contribute to increasing, respectively reducing, the deviations of prices from fundamentals.
  

{\bf Directional certificates.} In contrast to directional taxes where the \lq price tag\rq{} $c_{\rm Tax}$ for changing the MSD is exogenously fixed, directional certificates (that are themselves traded) endogenously assign a price to losses and gains of efficiency.\footnote{\Red The setup is analogous to emission trading schemes. A central authority sells a limited number of permits that allow a discharge of a specific quantity of a specific pollutant over a prespecified time period. Polluters are required to hold permits in an amount equal to their emissions. Polluters that want to increase their emissions must buy permits from others willing to sell them.} The regulator limits the supply of directional certificates (which could even be a negative total supply), and any bank $i$ needs to hold a number of certificates proportional to the size $h$ of the changes of its portfolio and the derivative of the MSD in direction $r$ of the changes. In particular, if changes in holdings and the corresponding directional derivatives are large -- indicating a significant increase of the MSD due to its actions -- a bank must hold many certificates. Conversely, if the bank's trades in assets decrease the MSD, the bank can go short on directional certificates. The key difference from directional taxes is that certificates are initially distributed to banks, and then traded between banks so that their unit price is determined in equilibrium.\footnote{\Red Banks which have stronger needs to trade in a direction that increases inefficiencies would be purchasing certificates from those whose trading motives are lower. Hence, certificates allow banks to internalize the externalities they are imposing on the system.} Certificates would thus cap the maximum inefficiency attainable in the market. 


}

\section{Conclusion}\label{sec:conclusion}

We developed a model to examine the ex-ante asset holdings which minimize market inefficiency in a systemic multi-asset economy. Price pressure arises in our model from the exogenous trading actions of banks which manage their portfolios to target specific leverage ratios. In the model, we quantify efficiency in terms of the mean squared deviation of fundamental versus market value capitalization. We find that inefficiencies are low if banks are not systemically significant, but become substantial if the overall systemic significance is high and banks are not sufficiently heterogeneous in systemic significance. We develop a procedure which constructs f-efficient holdings, and show that these depend on two key drivers, namely the banks' systemic significance and the first two moments of exogenous asset shocks. Our analysis reveals that increasing heterogeneity in banks' systemic significance moves f-efficient holdings closer to full diversification, while heterogeneity in the distribution (expectation and variance) of exogenous asset value shocks moves f-efficient holdings away from diversification. In balance sheet scenarios characterized by high risk and illiquidity, deviating from f-efficient holdings would result in high inefficiencies and large variance of market capitalization.


\bibliographystyle{plainnat}
\bibliography{DiversityBib}

\clearpage
\begin{appendices}
	\setcounter{page}{1}
	\renewcommand{\thesection}{EC.\arabic{section}}
	\appendixpage
	{\Red
	\section{Additional Details and Proofs}\label{app:proofs}

	\subsection{The Leverage Targeting Model}\label{CappLarsreview}

	{\bf Banking Sector.} The economy consists of $N$ banks, whose stylized balance sheets consist of assets, debt, and equity. Banks manage their balance sheets by buying or selling assets so to keep their leverage ratios (debt to equity ratios) at specified target levels. 
	Banks hold a portfolio of assets at time $0$. Then, price shocks occur, and banks purchase or sell assets to restore leverage. These actions impose pressure on prices, and as a result, the market value of bank portfolios at time 1 deviates from its fundamental value. 
	
	The quantity (number of units) of asset~$k$ held by bank $i$ at time $t$ is denoted by $Q^{ki}_t$. We use $Q^i_t = (Q^{1i}_t , Q^{2i}_t , \ldots , Q^{Ki}_t)^\top \in\mathbb{R}^{K}$ to denote the vector of bank $i$'s holdings at $t$, and $\nv Q:=(Q_0^{ki})_{k=1,\dots,K,i=1,\ldots,N}\in\mathbb{R}^{K\times N}$ to denote the matrix of banks' holdings at time zero. We write
	$
	A^i_t = (A^{1i}_t , A^{2i}_t , \ldots , A^{Ki}_t)^\top,
	$
	where $A^{ki}_t = P^k_t Q^{ki}_t$ is the market value of the $i$'th bank's holdings in asset $k$ at $t$. The total market value of the $i$'th bank is given by ${\bf 1}^\top A^i_t=\sum_{k=1}^K A^{ki}_t$. 
	
	Banks finance purchases by issuing debt. We use $D^i_t$ to denote the total amount of debt issued by bank~$i$ at time $t$, and assume that the interest rate on the debt is zero.\footnote{\Red Accounting for an exogenous nonzero interest rate would not qualitatively affect the results. Because our focus is on the price inefficiencies caused by banks' trading responses to shocks, we opt for a simpler model that highlights these effects.} The main behavioral assumption in the model is that each bank $i$ targets a fixed leverage ratio (debt to equity ratio) $\kappa^i$, i.e.,\footnote{\Red In practice, banks do not immediately revert to the target leverage. \cite{duarte2018fire} estimate the speed of leverage adjustment, and find that leverage adjustment speed is roughly constant until 2006, before increasing by over $50\%$ and spiking in 2008 due to the greater delevering through balance sheet contraction. Because we consider a one period snapshot of the economy, we may simply view the target leverage as a short-term target leverage.}
	
	\begin{equation} \label{eq:L}
	\frac{D^i_t}{{\bf 1}^\top A^i_t - D^i_t} =\kappa^i,\quad t=0,1,\quad i=1,\ldots,N.
	\end{equation}
	Each bank thus needs to adjust its debt and asset holdings in response to price changes.
 Each bank executes an exogenous trading strategy $\alpha^{i}\in\bbr^K$, which specifies how a change in the amount of debt is offset by purchases or sales of the different assets in the portfolio.
	 Hence, it holds that $\sum_{k=1}^K\alpha^{ki}=1$.      
	{\Red To be more specific, the fundamental cash-flow equation is given by:
	\begin{equation} \label{eq:CF}
	P^k_{1} \Delta Q^{ki} = \alpha^{ki} \Delta D^i, \quad k=1,\ldots,K,
	\end{equation}
	where $\Delta Q^{ki}:=Q^{ki}_1-Q^{ki}_0$ denotes the change in quantities from period 0 to period 1, and $\Delta D^i$ denotes the change in debt for bank $i$. The left-hand side is thus the change in value of bank $i$'s holdings of asset $k$, while the right-hand side is the change in debt needed to target the leverage, multiplied by~$\alpha^{ki}$. }
	For future purposes, let $\nv \alpha:=(\alpha^{ki})_{k=1,\dots,K,i=1,\ldots,N}\in\mathbb{R}^{K\times N}$ denote the trading strategy matrix, and let $L^i=\kappa^i/(1+\kappa^i)$. Hence, by the leverage equation \eqref{eq:L}, it always holds that $D^i_t =L^i \nv 1^\top A_t^i$ for $t=0,1$.
	
	The banks' demand curves, where $\Delta Q^{ki}:=Q^{ki}_1-Q^{ki}_0$ denotes the change in quantities from period 0 to period 1, admit an explicit expression.
	
	\begin{proposition}[\cite{capponi2015price}] \label{A:P:dc}
		The  incremental demand of bank $i$ for asset~$k$ is given by
		\begin{equation} \label{eq:dc}
		\Delta Q^{ki} = \alpha^{ki} \kappa^i Q_0^{i\top} \frac{ \Delta P}{P^k_{1}} .
		\end{equation}
	\end{proposition}

\begin{proof}
	This proof is obtained by specializing Proposition 1.1 in \cite{capponi2015price} to a setting with only one period, and assuming zero revenue shocks therein, i.e., $\Delta R^i=0$. 
	Writing the fundamental cash-flow equation~\eqref{eq:CF} in vector form yields
	\[
	\Diag(P_{1}) \Delta Q^i = \alpha^i \Delta D^i.
	\]
	Substituting for $D^i_t$ in the above equation the expressions for $L^i$, we obtain 
	\begin{align}\notag
		\Diag(P_{1}) \Delta Q^i &= \alpha^i \left( L_1^i \nv 1^\top A_1^i - L_0^i \nv 1^\top A_0^i  \right) \\
		&= \alpha^i \left( L_1^i \nv 1^\top (A_1^i - A_0^i) + (L_1^i-L_0^i) \nv 1^\top A_0^i  \right)
	\end{align}
	Rearranging the above expression leads to 
	\begin{equation} \label{eq:P:dc:1}
		\left( \Diag(P_{1}) - L^i \alpha^i P_{1}^\top \right) \Delta Q^i = \alpha^i L^i Q^{i\top}_0 \Delta P.
	\end{equation}
	The matrix multiplied by $\Delta Q^i$ can be inverted using the Sherman-Morrison formula. First, since $\nv 1^\top \alpha^i=1$, we have
	\[
	1-L^i P^\top_{1} \Diag(P_{1})^{-1} \alpha^i = 1-L^i\nv 1^\top \alpha^i = 1 - L^i \neq 0,
	\]
	so invertibility is guaranteed. The inverse is given by
	\[
	\Diag(P_{1})^{-1} + \kappa^i \Diag(P_{1})^{-1} \alpha^i P^\top_{1} \Diag(P_{1})^{-1},
	\]
	which simplifies to $\Diag(P_{1})^{-1} (\nv I + \kappa^i \alpha^i \nv 1^\top)$. From~\eqref{eq:P:dc:1} we therefore obtain
	\begin{align*}
		\Diag(P_{1})\Delta Q^i
		&= (\nv I + \kappa^i  \alpha^i \nv 1^\top) \left(\alpha^i L^i Q^{i\top}_0 \Delta P  \right) \\
		&= (1 + \kappa^i) \alpha^i \left( L^i Q^{i\top}_0 \Delta P \right),
	\end{align*}
	where the second equality uses the identity $(\nv I + \kappa^i \alpha^i \nv 1^\top)\alpha^i = (1 + \kappa^i) \alpha^i$, which follows from
	$\nv 1^\top\alpha^i=1$. Noting that $(1 + \kappa^i)L^i=\kappa^i$, we obtain 
	\begin{align*}
		\Diag(P_{1})\Delta Q^i &=  \kappa^i \alpha^i Q^{i\top}_0 \Delta P,\\
	\end{align*}
	and the stated expression~\eqref{eq:dc} follows.
\end{proof}
	
	\paragraph{The Nonbanking Sector.}
	Unlike banks, the nonbanking sector consists of institutions that are primarily equity funded (e.g., mutual funds, money market funds, insurances, and pension funds) and thus do not engage in leverage targeting. Nevertheless, they trade the same assets as the banking sector. 
	This gives rise to additional demand, which we refer to as the {\em nonbanking demand} and model it in a reduced form.  We assume that demand curves are decoupled across assets, i.e., the nonbanking demand for asset $k$ only depends on the price of asset $k$, and not on the prices of other assets.\footnote{\Red This modeling choice allows us to focus only on the price impact caused by the banks' needs of tracking their leverage requirements.} 
	
	At time $0$, the nonbanking sector holds a quantity $Q^{k,\,\rm nb}_0$ of asset $k$, and we write $A^{k,\,\rm nb}_0=P_0^k Q^{k,\,\rm nb}_0$ for the corresponding asset value.  
	Unlike the banking sector, whose demand function is upward sloping, the nonbanking sector has a downward sloping demand curve: it sells an asset if its price is above the fundamental value, and purchases an asset if its price is below its fundamental value. Hence, the nonbanking sector acts as the liquidity provider when there are shocks, and exerts a stabilizing force on the pressure imposed by banks. The demand for asset $k$ is given by
	\begin{equation}\label{eq:exd}
	\Delta Q^{k,\,\rm nb} \;= \; -\, \gamma^k  Q^{k,\,\rm nb}_0  \frac{ \Delta P^k   -   Z^k} {P_{1}^k},
	\end{equation}
	where $\gamma^k$ is a positive constant.  
	This choice of demand function admits the following interpretation. Assume no shock occurs, i.e., $Z^k=0$. Then
	\begin{equation}
	\frac{\Delta Q^{k,\,\rm nb}}{Q^{k,\,\rm nb}_0} = -\gamma^k \frac{\Delta P^k}{P^k_{1}}.
	\label{eq:deltaoutside}
	\end{equation}
	The parameter $\gamma^k$ can thus be interpreted as the elasticity of the nonbanking demand for asset~$k$, similar to $\kappa^i\alpha^{ki}$ in~\eqref{eq:dc} for the banking sector. We refer to $\gamma^k$ as the illiquidity characteristic of asset $k$. Unlike equation \eqref{eq:deltaoutside}, equation \eqref{eq:exd} includes the correction term $Z_k$, because nonbanking demand is due to deviations from fundamental values. 
}

	\subsection{Proofs of Section \ref{sec:model}}	
	
	\subsubsection{Proof of Proposition \ref{A:P:dyn}}
	This proof is obtained by specializing Proposition 2.1 in \cite{capponi2015price} to a setting with only one period, and assuming zero revenue shocks therein, i.e., $\Delta R^i=0$.
	First, we use the market-clearing condition $Q^{\rm nb}_t + \sum\nolimits_{i=1}^N Q^i_t = Q_{\rm tot},$ $t = 0,1,$ and then the expressions~\eqref{eq:dc} and~\eqref{eq:exd} for the demand functions to get
	\begin{align*}
	\boldsymbol{0} &= P_{1} \circ \Delta Q^{\rm nb} + \sum_{i=1}^N P_{1}\circ \Delta Q^i \\
	&= \Diag(\gamma\circ Q^{\rm nb}_0) \left( Z - \Delta P \right) + \sum_{i=1}^N \alpha^i \kappa^i Q^{i\top}_0\Delta P .
	\end{align*}
	Multiplying from the left by $\Diag(\gamma\circ Q^{\rm nb}_0)^{-1}$ and rearranging yields
	\begin{equation} \label{eq:P:price1}
	\left[ \Id - \sum_{i=1}^N \frac{\alpha^i}{\gamma\circ Q^{\rm nb}_0} \kappa^i Q^{i\top}_0\right]\Delta P
	= Z .
	\end{equation}
	The left-hand side is thus equal to $(\Id-\nv S)\Delta P$. We now simply multiply both sides of the equality \eqref{eq:P:price1} from the left by $(\Id-\nv S)^{-1}$ to arrive at the stated price change.


	\subsection{Proofs of Section \ref{sec:efficientholdings}}
	
	\subsubsection{Proof of Proposition \ref{prop:solutionmethod}}
	\begin{itemize}
		\item We first prove the direction (i)$\Rightarrow$(ii): Let $\nv Q^*\in\underset{\nv Q\in\mathcal{D}}{{\argmin}}\quad (\nv Q v)^\top \nv G (\nv Qv)$ be an f-efficient holding matrix. Since the objective function does only depend on the vector $\nv Q v\in\bbr^K$, the following statements are equivalent:
		\begin{equation}\label{eq:equivalentproblemy}
		\nv Q^*\in\underset{\nv Q\in\mathcal{D}}{{\argmin}}\quad (\nv Q v)^\top \nv G (\nv Qv)\;\Leftrightarrow\;\widetilde{y}:=\nv Q^* v\in\underset{y\in\widetilde{\mathcal{D}}}{{\argmin}}\quad y^\top \nv G y,
		\end{equation}
		where $\widetilde{\mathcal{D}}:=\{y\in\bbr^K\mid \exists\;\nv Q\in\mathcal{D}\;\text{s.t. } y=\nv Q v \}.$
		Since every vector of the form $y=\nv Q v$ for one $\nv Q\in\mathcal{D}$ satisfies $$\nv{1}_K^\top y=\nv{1}_K^\top\nv Q v=b^\top v,$$
		it obviously holds that $\widetilde{\mathcal{D}}\subseteq \mathcal{D}_y=\{y\in\bbr^K\mid \boldsymbol{1}_K^\top y=b^\top v\}$ yielding $\underset{y\in\widetilde{\mathcal{D}}}{\min}\quad y^\top\nv G y\ge \underset{y\in\mathcal{D}_y}{\min}\quad y^\top\nv G y.$ However, as shown by Proposition \ref{prop:existenceallocation}, if $y^*\in\mathcal{D}_y$ solves the aggregate problem, i.e., $y^*\in\underset{y\in\mathcal{D}_y}{{\argmin}}\quad y^\top \nv G y$, then there exists a matrix $ \widetilde{\nv Q}$ such that $y^*=\widetilde{\nv Q}v$, i.e., $y^*\in\widetilde{\mathcal{D}}$. Hence,
		\[\widetilde{y}^\top\nv Q\widetilde{y}= \underset{y\in\widetilde{\mathcal{D}}}{\min}\quad y^\top\nv G y = \underset{y\in\mathcal{D}_y}{\min}\quad y^\top\nv G y=(y^*)^\top\nv G y^*,  \]
		i.e. $\widetilde{y}=\nv Q^* v$ solves the aggregate problem \eqref{eq:yproblem}.
		\item Second, we prove the direction (ii)$\Rightarrow$ (i): Let $y^*=\nv Q^* v\in\underset{y\in\mathcal{D}_y}{{\argmin}}\quad y^\top \nv G y$ solve the aggregate problem. Since $\underset{y\in\widetilde{\mathcal{D}}}{\min}\quad y^\top\nv G y\ge \underset{y\in\mathcal{D}_y}{\min}\quad y^\top\nv G y$, with $\widetilde{\mathcal{D}}$ as defined in the previous step, and since $y^*=\nv Q^* v$ for some $\nv Q^*\in\mathcal{D}$, we have that $y^*=\nv Q^* v\in\underset{y\in\widetilde{\mathcal{D}}}{{\argmin}}\quad y^\top \nv G y$, which, finally, is equivalent to $\nv Q^*$ being f-efficient, see \eqref{eq:equivalentproblemy}.
	\end{itemize}
	
	\subsubsection{Proof of Lemma \ref{lem:aggregateholdings}}
	
	The KKT conditions for minimizing $y^\top {\nv G} y$ subject to $y\in\mathcal{D}_y$ read as
	\begin{equation}\label{eq:bigbankKKT}
	\underbrace{\begin{pmatrix}
		\nv G & \boldsymbol{1}_K\\
		\boldsymbol{1}^\top_K & 0\\
		\end{pmatrix}}_{=:\nv M} \cdot \begin{pmatrix}
	y\\\lambda
	\end{pmatrix}
	=
	\begin{pmatrix}
	\boldsymbol{0}_K\\
	b^\top v
	\end{pmatrix}.
	\end{equation}
	The inverse of the KKT matrix $\nv M$ is by blockwise inversion (see, e.g., Proposition 2.8.7 in \cite{bernstein2005matrix}) given by 
	\[	\nv M^{-1}=
	\begin{pmatrix}
	\nv G^{-1}-\nv G^{-1}\boldsymbol{1}_K(\boldsymbol{1}_K^\top \nv G^{-1}\boldsymbol{1}_K)^{-1}\boldsymbol{1}_K^\top \nv G^{-1} & \nv G^{-1}\boldsymbol{1}_K(\boldsymbol{1}_K^\top \nv G^{-1}\boldsymbol{1}_K)^{-1}\\
	(\boldsymbol{1}_K^\top \nv G^{-1}\boldsymbol{1}_K)^{-1}\boldsymbol{1}^\top_K \nv G^{-1} & -(\boldsymbol{1}_K^\top \nv G^{-1}\boldsymbol{1}_K)^{-1}\\
	\end{pmatrix}=:\begin{pmatrix}
	\nv A & \nv B\\
	\nv C & \nv D\\
	\end{pmatrix},  \]
	with $\nv A\in\bbr^{K\times K}$, $\nv B\in\bbr^{K\times 1}$, $\nv C\in\bbr^{1\times K}$ and $\nv D\in\bbr$. Note that this inverse matrix exists because, firstly, $\nv G$ is invertible because it is symmetric and positive definite, with inverse given by the Sherman-Morrison formula  
	\[\nv G^{-1}={\rm Diag}(\tfrac{1}{\sigma^2})- \frac{\tfrac{\mu}{\sigma^2}(\tfrac{\mu}{\sigma^2})^\top}{1+(\tfrac{\mu}{\sigma^2})^{\top} \mu} , \]
	and, secondly, $$\boldsymbol{1}_K^\top \nv G^{-1}\boldsymbol{1}_K=(\tfrac{1}{\sigma^2})^\top \nv{1}_K-\frac{(\nv{1}_K^\top\tfrac{\mu}{\sigma^2})^2}{1+(\tfrac{\mu}{\sigma^2})^\top \mu} \not=0,$$
	since with $c:=(\tfrac{1}{\sigma^2})^\top \nv{1}_K>0$
	\begin{align*}
	(\tfrac{1}{\sigma^2})^\top \nv{1}_K-\frac{(\nv{1}_K^\top\tfrac{\mu}{\sigma^2})^2}{1+(\tfrac{\mu}{\sigma^2})^\top \mu}=0\quad &\Leftrightarrow\quad  c(1+(\tfrac{\mu}{\sigma^2})^\top \mu)-(\nv{1}_K^\top\tfrac{\mu}{\sigma^2})^2=0\\
	&\Leftrightarrow\quad c+c\cdot (\tfrac{\mu}{\sigma})^\top(\tfrac{\mu}{\sigma})-(\tfrac{\mu}{\sigma})^\top\tfrac{1}{\sigma}(\tfrac{1}{\sigma})^\top\tfrac{\mu}{\sigma}=0\\
	&\Leftrightarrow\quad (\tfrac{\mu}{\sigma})^\top(\underbrace{\nv{I}-\tfrac{1}{c}\cdot\tfrac{1}{\sigma}(\tfrac{1}{\sigma})^\top}_{=:\nv{E}})\tfrac{\mu}{\sigma}=-1,
	\end{align*}
	and this cannot be fulfilled for any $\tfrac{\mu}{\sigma}\in\mathbb{R}^K$ since $\nv{E}\in\mathbb{R}^{K\times K}$ is a positive semidefinite matrix due to its only eigenvalues $0$ and $1$ (cf. \cite{dattorro2005convex}, Appendix B.3).
	
	Hence, the f-efficient solution to problem \eqref{eq:yproblem} is derived from multiplying both sides of equation \eqref{eq:bigbankKKT} by $\nv M^{-1}$ yielding
	\begin{align*}
	y^{*}&=\nv A\cdot\nv 0_K+ \nv B\cdot (b^\top v) =b^\top v\cdot \nv G^{-1}\boldsymbol{1}_K(\boldsymbol{1}_K^\top \nv G^{-1}\boldsymbol{1}_K)^{-1}
	=\frac{b^\top v}{\nv{1}_K^\top z} z,
	\end{align*}
	where $z:=\nv G^{-1}\nv{1}_K$.
	
	\subsubsection{Proof of Proposition \ref{prop:existenceallocation}}
	We define the matrix
	\begin{equation}\label{eq:Qp}
	\begin{aligned}
	\nv Q^p&:=\left(\tfrac{1}{v_2-v_1}[v_2 q-y^*-(\sum\limits_{i=3}^{N}(v_2-v_i)b_i)e_1^K],\;\tfrac{1}{v_2-v_1}[y^*-v_1 q-(\sum\limits_{i=3}^{N}(v_i-v_1)b_i)e_1^K],\; b_3 e_1^K,\;\ldots,\;b_N e_1^K\right)\\&\in\bbr^{K\times N},
	\end{aligned}
	\end{equation}	
	where $e_1^K=(1,0,\ldots,0)^\top\in\bbr^K$. Without loss of generality, we have assumed here that $v_2\not=v_1$ (cf. Assumption \ref{ass:v}).
	This matrix satisfies $\nv Q^p\in\mathcal{D}$ as well as $y^*=\nv Q^p v$, for $y^*=\tfrac{b^\top v}{\nv{1}_K^\top z} z$, the solution of the aggregate problem \eqref{eq:yproblem} derived in Lemma \ref{lem:aggregateholdings}: 
	\begin{itemize}
		\item \underline{$y^*=\nv Q^p v$:} It holds:
		\begin{align*}
		(\nv Q^p v)_1&=\frac{1}{v_2-v_1}\left(v_1v_2q_1-y_1^*v_1-v_1\sum\limits_{i=3}^{N}(v_2-v_i)b_i+y_1^*v_2-v_1v_2q_1-v_2\sum\limits_{i=3}^{N} (v_i-v_1)b_i \right)\\&\quad +\sum\limits_{i=3}^{N} b_i v_i =\frac{1}{v_2-v_1}\left((v_2-v_1)y_1^*-(v_2-v_1)\sum\limits_{i=3}^{N}b_i v_i \right)+\sum\limits_{i=3}^{N} b_i v_i=y_1^*,
		\end{align*}
		and for $k=2,\ldots,K$:
		\begin{align*}
		(\nv Q^p v)_k=\frac{1}{v_2-v_1}\left(v_1v_2q_k-y_k^*v_1+y_k^*v_2-v_1v_2q_k \right)=y_k^*.
		\end{align*}
		\item \underline{$\nv Q^p\nv{1}_N=q$:} We obtain
		\begin{align*}
		(\nv Q^p\nv{1}_K)_1&=\frac{1}{v_2-v_1}\left(v_2q_1-y_1^*-\sum\limits_{i=3}^{N}(v_2-v_i)b_i+y_1^*-v_1q_1-\sum\limits_{i=3}^{N} (v_i-v_1)b_i \right)+\sum\limits_{i=3}^{N} b_i\\
		&=\frac{1}{v_2-v_1}\left((v_2-v_1)q_1+(v_2-v_1)\sum\limits_{i=3}^{N} b_i \right)+\sum\limits_{i=3}^{N} b_i=q_1,
		\end{align*}
		and for $k=2,\ldots,K$:
		\begin{align*}
		(\nv Q^p\nv{1}_N)_k=\frac{1}{v_2-v_1}\left(v_2q_k-y_k^*+y_k^*-v_1q_k \right)=q_k.
		\end{align*}
		\item \underline{$\nv{1}_K^\top\nv Q^p=b^\top$:} Obviously, it is $(\nv{1}_K^\top \nv Q^p)_i=b_i$, for $i=3,\ldots,N$. For the first and second entry, it holds that
		\begin{align*}
		(\nv{1}_K^\top\nv Q^p)_1&=\frac{1}{v_2-v_1}\left(v_2\nv{1}_K^\top q-\nv{1}_K^\top y^*-\sum\limits_{i=3}^{N}(v_2-v_i)b_i \right)\\
		&=\frac{1}{v_2-v_1}\left(v_2\sum\limits_{i=1}^{N}b_i-\sum\limits_{i=1}^{N}b_iv_i-\sum\limits_{i=3}^{N}v_2b_i+\sum\limits_{i=3}^{N}b_iv_i \right)\\
		&=\frac{1}{v_2-v_1}(v_2(b_1+b_2)-(v_1b_1+v_2b_2))=b_1,
		\end{align*}
		where in the first step, we have used that $\sum\nolimits_{k=1}^{K} q_k=T=\sum\nolimits_{i=1}^N b_i$ and that $\nv{1}_K^\top y^*=b^\top v$. Using the same arguments, we obtain
		\begin{align*}
		(\nv{1}_K^\top\nv Q^p)_2&=\frac{1}{v_2-v_1}\left(\nv{1}_K^\top y^*-v_1\nv{1}_K^\top q-\sum\limits_{i=3}^{N}(v_i-v_1)b_i \right)\\
		&=\frac{1}{v_2-v_1}\left(\sum\limits_{i=1}^{N} b_i-v_1\sum\limits_{i=1}^{N}b_iv_i-\sum\limits_{i=3}^{N}b_iv_i+\sum\limits_{i=3}^{N}v_1b_i \right)\\
		&=\frac{1}{v_2-v_1}(b_1v_1+b_2v_2-(v_1b_1+v_1b_2))=b_2,
		\end{align*}
		which completes the proof.
	\end{itemize}

	\begin{remark}\label{rem:FullSolutionSpace}
		Due to our two-step solution method proven in Proposition \ref{prop:solutionmethod}, finding an f-efficient holding matrix $\nv{Q}^*\in\bbr^{K\times N}$ is equivalent to solving the linear system
		\begin{align}\label{eq:LSEAllocationMatrix}
		\underbrace{\begin{pmatrix}
			\boldsymbol{1}_K^\top & 0 & \cdots & 0\\
			0 & \nv{1}_K^\top & \ddots & \vdots \\
			\vdots & \ddots & \ddots &0\\
			0 & \cdots & 0 & \nv{1}_K^\top\\
			\nv{I}_K & \nv{I}_K & \cdots & \nv{I}_K\\
			v_1\nv{I}_K & v_2\nv{I}_K&\cdots & v_N\nv{I}_K
			\end{pmatrix}}_{=:\nv F\in\mathbb{R}^{2K+N\times KN}}	{\rm vec}(\nv Q^*) =\begin{pmatrix}
		b\\q\\y^*
		\end{pmatrix},
		\end{align}
		where ${\rm vec}(\nv Q^*)\in\bbr^{KN}$ denotes the vectorized version of the matrix $\nv Q^*$, obtained by stacking its columns on top of one another. The null space corresponding to the matrix $\nv F$ is spanned by the $(K-1)(N-2)$ column vectors of the matrix (w.l.o.g. $v_1\not=v_2$, cf. Assumption \ref{ass:v})
		\begin{equation}\label{eq:nullspace}
		\nv{O}:=
		\begin{pmatrix}
		\tfrac{v_2-v_3}{v_2-v_1}\begin{psmallmatrix}
		-\nv{1}_{K-1}^\top\\
		\nv{I}_{K-1}
		\end{psmallmatrix}
		&\tfrac{v_2-v_4}{v_2-v_1}\begin{psmallmatrix}
		-\nv{1}_{K-1}^\top\\
		\nv{I}_{K-1}
		\end{psmallmatrix}
		&\cdots &
		\tfrac{v_2-v_N}{v_2-v_1}\begin{psmallmatrix}
		-\nv{1}_{K-1}^\top\\
		\nv{I}_{K-1}
		\end{psmallmatrix}\\
		\tfrac{v_3-v_1}{v_2-v_1}\begin{psmallmatrix}
		-\nv{1}_{K-1}^\top\\
		\nv{I}_{K-1}
		\end{psmallmatrix}
		&
		\tfrac{v_4-v_1}{v_2-v_1}\begin{psmallmatrix}
		-\nv{1}_{K-1}^\top\\
		\nv{I}_{K-1}
		\end{psmallmatrix}
		&\cdots& \tfrac{v_N-v_1}{v_2-v_1}\begin{psmallmatrix}
		-\nv{1}_{K-1}^\top\\
		\nv{I}_{K-1}
		\end{psmallmatrix}\\
		\begin{psmallmatrix}
		\nv{1}_{K-1}^\top\\
		-\nv{I}_{K-1}
		\end{psmallmatrix}
		& \nv{0}_{K\times K-1}
		&\cdots & \nv{0}_{K\times K-1}\\
		\nv{0}_{K\times K-1}
		& 	\begin{psmallmatrix}
		\nv{1}_{K-1}^\top\\
		-\nv{I}_{K-1}
		\end{psmallmatrix}&
		\ddots &
		\vdots\\
		\vdots & \ddots & \ddots &\nv{0}_{K\times K-1}\\
		\nv{0}_{K\times K-1} & \cdots & \nv{0}_{K\times K-1} & \begin{psmallmatrix}
		\nv{1}_{K-1}^\top\\
		-\nv{I}_{K-1}
		\end{psmallmatrix}
		\end{pmatrix}\in\bbr^{KN\times(K-1)(N-2)}.
		\end{equation}
		The fact that the linearly independent column vectors of $\nv O$ lie in the null space of $\nv F$ is easily checked via direct calculation; the fact that the dimension of the null space is equal to $(K-1)(N-2)$ follows from a rank-nullity argument given in the proof of Theorem \ref{thm:efficientholdings} b) below.
		Hence, we are able to fully characterize the set of f-efficient holding matrices as
		\[\left\{\nv Q^*\in\bbr^{K\times N}\mid {\rm vec}(\nv Q^*)={\rm vec}(\nv Q^p)+\sum\limits_{j=1}^{(K-1)(N-2)}\lambda_j C_j(\nv O) ,\,\lambda_1,\ldots,\lambda_{(K-1)(N-2)}\in\bbr\right\}, \]
		with particular solution $\nv Q^p$ as defined in \eqref{eq:Qp} and where $C_j(\nv O)$, for $j=1,\ldots,(K-1)(N-2)$, denotes the $j$-th column vector of the null space matrix $\nv O$.
	\end{remark}
	
	\subsubsection{Proof of Theorem \ref{thm:efficientholdings}}
	\begin{itemize}
		\item[a)] The existence of an f-efficient holding matrix for every $N,K\ge 2$ directly follows from Propositions \ref{prop:existenceallocation} and \ref{prop:solutionmethod}. Let $y^*=\tfrac{b^\top v}{\nv{1}_K^\top z} z$, the solution of the aggregate problem \eqref{eq:yproblem} derived in Lemma \ref{lem:aggregateholdings} with $z=\nv{G}^{-1}\nv{1}_K$. The mean squared deviation under an f-efficient holding matrix $\nv Q^*$ is given by:
		\begin{align*}
		MSD(\nv Q^*)&=(\nv Q^*v)^\top \nv G (\nv Q^* v)= (y^*)^\top\nv{G}y^*\\
		&=\frac{(b^\top v)^2}{(\nv{1}_K^\top z)^2} z^\top\nv{G}z=\frac{(b^\top v)^2}{(\nv{1}_K^\top z)^2} z^\top \nv{G}\nv G^{-1}\nv{1}_K\\
		&=\frac{(b^\top v)^2}{(\nv{1}_K^\top z)^2} (z^\top\nv{1}_K)=\frac{(b^\top v)^2}{\nv{1}_K^\top z}=\frac{(b^\top v)^2}{\nv{1}_K^\top \nv{G}^{-1}\nv{1}_K}.
		\end{align*}
		\item[b)] As outlined above in Remark \ref{rem:FullSolutionSpace}, finding an f-efficient holding matrix $\nv{Q}^*\in\bbr^{K\times N}$ is equivalent to solving the linear system
		\eqref{eq:LSEAllocationMatrix}.
		We have the following result: The rank of the matrix $\nv F$ is equal to $2K+N-2$.
		This can be proven via standard Gaussian elimination. First, resorting the rows of $\nv F$ and adding the new first $K$ rows multiplied by $-v_1$ to the second $K$ rows yields
		\begin{align*}
		\begin{pmatrix}
		\nv{I}_K & \nv{I}_K & \cdots & \nv{I}_K\\
		v_1\nv{I}_K & v_2\nv{I}_K&\cdots & v_N\nv{I}_K\\
		\boldsymbol{1}_K^\top & 0 & \cdots & 0\\
		0 & \nv{1}_K^\top & \ddots & \vdots \\
		\vdots & \ddots & \ddots &0\\
		0 & \cdots & 0 & \nv{1}_K^\top\\
		\end{pmatrix}
		\rightarrow
		\begin{pmatrix}
		\nv{I}_K & \nv{I}_K & \cdots & \nv{I}_K\\
		\nv{0}_{K\times K} & (v_2-v_1)\nv{I}_K&\cdots & (v_N-v_1)\nv{I}_K\\
		\boldsymbol{1}_K^\top & 0 & \cdots & 0\\
		0 & \nv{1}_K^\top & \ddots & \vdots \\
		\vdots & \ddots & \ddots &0\\
		0 & \cdots & 0 & \nv{1}_K^\top\\
		\end{pmatrix}.
		\end{align*}
		Numbering the rows in this last matrix as $r_1,\ldots,r_{2K+N}$, we observe:
		\begin{align*}
		r_{2K+1}=\sum\limits_{i=1}^{K}r_i-\sum\limits_{i=2K+2}^{2K+N} r_i
		\end{align*}
		and
		\begin{align*}
		r_{2K+2}=\frac{1}{v_2-v_1}\left(\sum\limits_{i=K+1}^{2K}r_i-\sum\limits_{i=2K+3}^{2K+N}(v_{i-2K}-v_1)r_i \right),
		\end{align*}
		where, as above, w.l.o.g. $v_2\not=v_1$ (cf. Assumption \ref{ass:v}). Hence, these two rows can be eliminated from the matrix yielding the row-echelon form:
		\begin{equation*}
			\begin{pmatrix}
			\nv{I}_K & \nv{I}_K & \nv{I}_K & \nv{I}_K& \cdots & \nv{I}_K\\
			\nv{0}_{K\times K} & (v_2-v_1)\nv{I}_K&(v_3-v_1)\nv{I}_K & (v_4-v_1)\nv{I}_K&\cdots&(v_N-v_1)\nv{I}_K\\
			\boldsymbol{0}_K^\top & \boldsymbol{0}_K^\top & \cdots & &&\boldsymbol{0}_K^\top\\
			\boldsymbol{0}_K^\top & \nv{0}_K^\top & \cdots & &&\boldsymbol{0}_K^\top \\
			\boldsymbol{0}_K^\top & \boldsymbol{0}_K^\top & \boldsymbol{1}_K^\top &0&\cdots&0\\
			0 & \cdots & 0 & \nv{1}_K^\top&\ddots &\vdots\\
			\vdots & \ddots  &\vdots &\ddots &\ddots&0\\
			0&\cdots&0&\cdots&0&\nv{1}_K^\top
			\end{pmatrix},
		\end{equation*}	
		which shows that the matrix $\nv F$ possesses the rank $2K+N-2$.	
		The rank-nullity theorem now gives rise to the dimension of the null space (a basis is given by the column vectors of the matrix $\nv O$ in Remark \ref{rem:FullSolutionSpace} above):
		\[KN={\rm rank}(\nv F)+{\rm null}(\nv F)\;\Rightarrow\; {\rm null}(\nv F)=KN-(2K+N-2)=(N-2)(K-1). \]
		and, hence, the dimension of the null space is zero, i.e., the solution $\nv Q^p$ stated in the proof of Proposition \ref{prop:existenceallocation} is unique, if and only if either $K=1$ or $N=2$; i.e., if we additionally assume that $K\ge 2$, then this is equivalent to $N=2$. Finally, the formula for $\nv Q^{N=2}$ directly follows from the definition of $\nv Q^p$ in \eqref{eq:Qp}.
	\end{itemize}
	
	\subsubsection{Proof of Theorem \ref{thm:diversification}}
	Recall that $\nv Q^{\rm diversified}=\tfrac{1}{T}qb^\top=\tfrac{1}{\nv{1}_K^\top q}qb^\top$. Thus, it holds that
	\begin{align*}
	& y^*= \nv Q^{\rm diversified} v\\
	&\Leftrightarrow \quad \frac{b^\top v}{\nv{1}_K^\top z}\cdot z = \frac{1}{\nv{1}_K^\top q}q(b^\top v)\\
	&\overset{b^\top v\not=0}{\Leftrightarrow} \quad \frac{1}{\nv{1}_K^\top z} z = \frac{1}{\nv{1}_K^\top q}q,
	\end{align*}
	i.e., if and only if $q$ and $z=\nv{G}^{-1}\nv{1}_K$ are linearly dependent.
	
	\subsubsection{Proof of Corollary \ref{cor:diversification}}
	\begin{itemize}
		\item For the proof of statements a) and c), we first observe that if $q_1=\ldots=q_K$, then 
		\begin{align*}
		&\frac{1}{\nv{1}_K^\top z} z = \frac{1}{\nv{1}_K^\top q}q\quad\Leftrightarrow\quad \frac{1}{\nv{1}_K^\top z} z = \frac{1}{q_1\cdot K}q_1\cdot\nv{1}_K\\
		&\Leftrightarrow\quad  z =\frac{1}{K} \cdot\nv{1}_K\nv{1}_K^\top z\quad \Leftrightarrow\quad K\cdot z = (\nv{1}_K\nv{1}_K^\top) z,
		\end{align*}
		i.e., that $z$ is an eigenvector to the eigenvalue $K$ of the all-one matrix $\nv{1}_{K\times K}=\nv{1}_K\nv{1}_K^\top$. This eigenvector is given as $z=c\cdot\nv{1}_K$ for a constant $c\in\bbr$. Hence, $z_k=z_\ell$ for all $k,\ell=1,\ldots,K$ is equivalent to $\nv Q^{\rm diversified}$ being f-efficient under the assumption that $q_1=\ldots=q_K$.
		
		To prove part a), we now set $\sigma_1^2=\ldots=\sigma_K^2$. According to the proof of Lemma \ref{lem:aggregateholdings}, the vector $z$ is in this situation given by
		\begin{align*}
		z=\nv{G}^{-1}\nv{1}_K=\frac{1}{\sigma_1^2}\cdot\nv{1}_K-\frac{(\tfrac{1}{\sigma_1^2})^2\cdot\mu^\top\nv{1}_K}{1+\tfrac{1}{\sigma_1^2\cdot\mu^\top\mu}}\cdot\mu.
		\end{align*}
		This means that the condition $z_k=z_\ell$ for all $k,\ell=1,\ldots,K$ is equivalent to 
		\begin{align*}
		&z_k=\frac{1}{\sigma_1^2}-\frac{(\tfrac{1}{\sigma_1^2})^2\cdot\sum\nolimits_{j=1}^K \mu_j}{1+\tfrac{1}{\sigma_1^2}\cdot\sum\nolimits_{j=1}^{K}\mu_j^2}\cdot \mu_k \overset{!}{=}\frac{1}{\sigma_1^2}-\frac{(\tfrac{1}{\sigma_1^2})^2\cdot\sum\nolimits_{j=1}^K \mu_j}{1+\tfrac{1}{\sigma_1^2}\cdot\sum\nolimits_{j=1}^{K}\mu_j^2}\cdot \mu_\ell=z_\ell\\
		&\Leftrightarrow\quad \frac{(\tfrac{1}{\sigma_1^2})^2\cdot\sum\nolimits_{j=1}^K \mu_j}{1+\tfrac{1}{\sigma_1^2}\cdot\sum\nolimits_{j=1}^{K}\mu_j^2}\cdot \mu_k =\frac{(\tfrac{1}{\sigma_1^2})^2\cdot\sum\nolimits_{j=1}^K \mu_j}{1+\tfrac{1}{\sigma_1^2}\cdot\sum\nolimits_{j=1}^{K}\mu_j^2}\cdot \mu_\ell\\
		&\Leftrightarrow\quad (\sum\nolimits_{j=1}^K \mu_j)\cdot\mu_k=(\sum\nolimits_{j=1}^K\mu_j)\cdot\mu_\ell,
		\end{align*}
		for all $k,\ell=1,\ldots,K$, i.e., $\mu_k=\mu_{\ell}$ or $\sum\nolimits_{j=1}^K\mu_j=0$. Hence, if $\mu_1=\ldots=\mu_{K-1}$, and $\mu_K=\mu_1+\varepsilon$, then, finally, $\nv{Q}^{\rm diversified}$ being f-efficient is equivalent to either $\varepsilon=0$, yielding $\mu_k=\mu_{\ell}$ for all $k,\ell=1,\ldots,K$, or $\varepsilon=-K\mu_1$, yielding $\sum\nolimits_{j=1}^K\mu_j=0$.
		
		To prove part c), we set $\mu_1=\ldots=\mu_K$, which (see the proof of Lemma \ref{lem:aggregateholdings}) leads to the vector $z$ given as:
		\begin{align*}
		z=\nv{G}^{-1}\nv{1}_K=\frac{1}{\sigma^2}-\frac{\mu_1^2\cdot(\tfrac{1}{\sigma^2})^\top\nv{1}_K}{1+\mu_1^2\cdot(\tfrac{1}{\sigma^2})^\top\nv{1}_K}\cdot \frac{1}{\sigma^2}.
		\end{align*}
		Thus, the condition $z_k=z_\ell$ for all $k,\ell=1,\ldots,K$ reads as
		\begin{align*}
		&z_k=\left(1-\frac{\mu_1^2\cdot\sum\nolimits_{j=1}^K\tfrac{1}{\sigma_j^2} }{1+\mu_1^2\cdot\sum\nolimits_{j=1}^K\tfrac{1}{\sigma_j^2}}\right)\frac{1}{\sigma_k^2}\overset{!}{=} \left(1-\frac{\mu_1^2\cdot\sum\nolimits_{j=1}^K\tfrac{1}{\sigma_j^2} }{1+\mu_1^2\cdot\sum\nolimits_{j=1}^K\tfrac{1}{\sigma_j^2}}\right)\frac{1}{\sigma_\ell^2}=z_\ell\\
		&\Leftrightarrow\quad \sigma_k^2=\sigma_\ell^2
		\end{align*}
		for all $k,\ell=1,\ldots,K$. Hence, if $\sigma_1=\ldots=\sigma_{K-1}$ and $\sigma_K=\sigma_1+\varepsilon$, this condition is fulfilled if and only if $\varepsilon=0$.
		
		\item For the remaining proof of part b), observe that if $\mu_1=\ldots=\mu_K$ and $\sigma_1^2=\ldots=\sigma_K^2$, then 
		\[z=\nv{G}^{-1}\nv{1}_K=\frac{1}{\sigma_1^2}\cdot\nv{1}_K-\frac{(\tfrac{\mu_1}{\sigma_1^2})^2\cdot K}{1+\tfrac{\mu_1^2}{\sigma_1^2}K}\cdot\nv{1}_K, \]
		in particular: $z_1=\ldots=z_K$. Hence:
		\begin{align*}
		&\frac{1}{\nv{1}_K^\top z} z = \frac{1}{\nv{1}_K^\top q}q\quad\Leftrightarrow\quad \frac{1}{z_1\cdot K} z_1\cdot\nv{1}_K =\frac{1}{\nv{1}_K^\top q}q\\
		&\Leftrightarrow\quad \nv{1}_K\nv{1}_K^\top q=K\cdot q
		\end{align*}
		which is equivalent to $q$ being an eigenvector to the eigenvalue $K$ of the all-one matrix $\nv{1}_{K\times K}=\nv{1}_K\nv{1}_K^\top$, i.e., $q_1=\ldots=q_K$. Hence, if $q_1=\ldots=q_{K-1}$ and $q_K=q_1+\varepsilon$, this condition is equivalent to $\varepsilon=0$, which completes the proof.
		
	\end{itemize}

	\subsection{Proofs of Section \ref{sec:casestudies}}
	
	\subsubsection{Proof of Lemma \ref{lem:sigma2}}
	First, we derive the explicit formulas for $\nv Q^{2\times 2}_{11}$ and $d(\nv Q^{2\times 2})$ in the special case $N=K=2$, $q_1=q_2=b_1=b_2=x$. Note that, here, the unique f-efficient holding matrix is given by
	\begin{align*}
	\nv Q^{2\times2}=\frac{1}{v_2-v_1}\begin{pmatrix}
	v_2 q_1-y_1^* & y_1^*-v_1 q_1\\
	v_2 q_2-y_2^* & y_2^*-v_1 q_2\\
	\end{pmatrix}
	=\frac{x}{(v_2-v_1)(z_1+z_2)}\begin{pmatrix}
	{v_2 z_2-v_1 z_1} & {v_2 z_1-v_1 z_2}\\
	{v_2 z_1-v_1 z_2} &{v_2 z_2-v_1 z_1}\\
	\end{pmatrix},
	\end{align*}
	due to $y^*=x\cdot\tfrac{v_1+v_2}{z_1+z_2}\cdot z$. Since, $\nv Q^{\rm diversified}=\tfrac{x}{2}\cdot\nv{1}_{2\times 2}$, it holds that 
	\begin{align*}
	\nv Q^{2\times 2}_{11}-\nv Q^{\rm diversified}_{11}&=x\cdot\frac{2v_2 z_2-2v_1 z_1-(v_2-v_1)(z_1+z_2)}{2(v_2-v_1)(z_1+z_2)}=x\cdot\frac{v_2z_2-v_1z_1-v_2z_1+v_1z_2}{2(v_2-v_1)(z_1+z_2)}\\
	&=x\cdot\frac{(z_2-z_1)(v_1+v_2)}{2(v_2-v_1)(z_1+z_2)}=\nv Q^{2\times 2}_{22}-\nv Q^{\rm diversified}_{22},
	\end{align*}
	and
	\begin{align*}
	\nv Q^{2\times 2}_{12}-\nv Q^{\rm diversified}_{12}&=x\cdot\frac{2v_2 z_1-2v_1 z_2-(v_2-v_1)(z_1+z_2)}{2(v_2-v_1)(z_1+z_2)}=x\cdot\frac{v_2z_1-v_1z_2+v_1z_1-v_2z_2}{2(v_2-v_1)(z_1+z_2)}\\
	&=x\cdot\frac{(z_1-z_2)(v_1+v_2)}{2(v_2-v_1)(z_1+z_2)}=\nv Q^{2\times 2}_{21}-\nv Q^{\rm diversified}_{21}.
	\end{align*}
	Hence,
	\begin{align}\label{eq:dQvz}
	d(\nv Q^{2\times 2})=\|\nv Q^{2\times 2}-\nv Q^{\rm diversified}\|_F=\sqrt{4\cdot x^2\cdot\frac{(z_1-z_2)^2(v_1+v_2)^2}{4(v_2-v_1)^2(z_1+z_2)^2}}=x\cdot\sqrt{\frac{(z_1-z_2)^2(v_1+v_2)^2}{(v_2-v_1)^2(z_1+z_2)^2}}.
	\end{align}
	Moreover, a direct calculation of $z=\nv{G}^{-1}\nv{1}_K$ in the 2-by-2-case shows that
	\begin{align}\label{eq:Q2by2msv}
	\nv Q^{2\times2}_{11}=x\cdot\frac{v_2(\mu_1^2+\sigma_1^2-\mu_1\mu_2)-v_1(\mu_2^2+\sigma_2^2-\mu_1\mu_2)}{(v_2-v_1)(\mu_1^2+\mu_2^2-2\mu_1\mu_2+\sigma_1^2+\sigma_2^2)},
	\end{align}
	and that
	\begin{align}\label{eq:dQ2msv}
	d(\nv Q^{2\times 2})=x\cdot\sqrt{\frac{(\mu_1^2-\mu_2^2+\sigma_1^2-\sigma_2^2)^2(v_1+v_2)^2}{(v_2-v_1)^2(\mu_1^2+\mu_2^2-2\mu_1\mu_2+\sigma_1^2+\sigma_2^2)^2}}.
	\end{align}
	\begin{itemize}
		\item[a)] Under the conditions $q_1=q_2=x$ and $\mu_1=\mu_2$, it holds that $d(\nv Q^{2\times 2})=0$ as a function of $\sigma_1$, if and only if $\sigma_1^2=\sigma_2^2$ (cf. Corollary \ref{cor:diversification}), and it is strictly positive everywhere else. Hence, we can equivalently analyze the monotonicity behavior of $d(\nv Q^{2\times 2})^2$. According to \eqref{eq:dQ2msv}, its derivative with respect to $\sigma_1$ is given by
		\[\frac{\partial}{\partial\sigma_1} d(\nv Q^{2\times 2})^2 = \frac{8\sigma_1\sigma_2^2(\sigma_1^2-\sigma_2^2)(v_1+v_2)^2 x^2}{(\sigma_1^2+\sigma_2^2)^3(v_1-v_2)^2}\begin{cases}
		<0,&\quad\text{if $\sigma_1^2<\sigma_2^2$,}\\
		=0,&\quad \text{if $\sigma_1^2=\sigma_2^2$,}\\
		>0,&\quad \text{if $\sigma_1^2>\sigma_2^2$,}
		\end{cases} \]
		for $\sigma_1>0$ under the given assumptions, which proves the statement.
		
		\item[b)] Now assume $x>0$ and $v_2>v_1>0$. The derivative of \eqref{eq:Q2by2msv} with respect to $\sigma_1$ is given by
		\begin{align*}
		\frac{\partial \nv Q^{2\times 2}_{11}}{\partial \sigma_1}=x\cdot\frac{2\sigma_1(\mu_2^2+\sigma_2^2-\mu_1\mu_2)(v_1+v_2)}{(v_2-v_1)(\mu_1^2+\mu_2^2-2\mu_1\mu_2+\sigma_1^2+\sigma_2^2)^2}\overset{\mu_1=\mu_2}{=} x\cdot\frac{2\sigma_1\sigma_2^2(v_1+v_2)}{(\sigma_1^2+\sigma_2^2)^2(v_2-v_1)}>0,
		\end{align*}
		for $\sigma_1>0$ under the given assumptions.
		
	\end{itemize}
	\subsubsection{Proof of Lemma \ref{lem:mu}}
	\begin{itemize}
		\item[a)] As in the proof of Lemma \ref{lem:sigma2}, $d(\nv Q^{2\times 2})$ as a function of $\mu_2$ is non-negative and, under the given assumptions, strictly positive except for the case $\mu_2=\mu_1$ (cf. Corollary \ref{cor:diversification}). Hence, we can again equivalently analyze the monotonicity behavior of the squared distance $d(\nv Q^{2\times 2})^2$. Its derivative with respect to $\mu_2$ is given by
		\begin{align*}
		\frac{\partial}{\partial\mu_2}d(\nv Q^{2\times 2})^2&=x^2\cdot\frac{4(\mu_1^2-\mu_2^2+\sigma_1^2-\sigma_2^2)(-2\mu_2\sigma_1^2+\mu_1(\mu_1^2+\mu_2^2-2\mu_1\mu_2+\sigma_1^2-\sigma_2^2))(v_1+v_2)^2}{(\mu_1^2-2\mu_1\mu_2+\mu_2^2+\sigma_1^2+\sigma_2^2)^3(v_2-v_1)^2}\\
		&\overset{\sigma_1^2=\sigma_2^2,\mu_1=0}{=} x^2\cdot\frac{8\mu_2^3\sigma_1^2(v_1+v_2)^2}{(\mu_2^2+2\sigma_1^2)^3(v_2-v_1)^2}\begin{cases}
		<0,&\quad\text{if $\mu_2<0$,}\\
		=0,&\quad \text{if $\mu_2=0$,}\\
		>0,&\quad \text{if $\mu_2>0$,}
		\end{cases}
		\end{align*}
		which proves the lemma.
		
		\item[b)] The derivative of \eqref{eq:Q2by2msv} with respect to $\mu_2$ is given by
		\begin{align*}
		\frac{\partial}{\partial\mu_2}\nv Q^{2\times 2}_{11}&=x\cdot\frac{(-2\mu_2\sigma_1^2+\mu_1(\mu_1^2+\mu_2^2-2\mu_1\mu_2+\sigma_1^2-\sigma_2^2))(v_1+v_2)}{(\mu_1^2+\mu_2^2-2\mu_1\mu_2+\sigma_1^2+\sigma_2^2)^2(v_2-v_1)}\\
		&\overset{\sigma_1^2=\sigma_2^2,\mu_1=0}{=} x\cdot\frac{-2\mu_2\sigma_1^2(v_1+v_2)}{(\mu_2^2+2\sigma_1^2)^2(v_2-v_1)}\begin{cases}
		>0,&\quad\text{if $\mu_2<0$,}\\
		=0,&\quad \text{if $\mu_2=0$,}\\
		<0,&\quad \text{if $\mu_2>0$,}
		\end{cases}
		\end{align*}
		under the assumptions $v_2>v_1>0$, $x>0$, $\sigma_1^2=\sigma_2^2$ and $\mu_1=0$.

	\end{itemize}
	
	\subsubsection{Proof of Lemma \ref{lem:v}}
	The derivative of \eqref{eq:dQvz} with respect to $v_2$ is given as
	\begin{align*}
	\frac{\partial}{\partial v_2}d(\nv Q^{2\times 2})=x\cdot\frac{2v_1}{(v_1+v_2)(v_1-v_2)}\cdot\sqrt{\frac{(v_1+v_2)^2(z_1-z_2)^2}{(v_2-v_1)^2(z_1+z_2)^2}}.
	\end{align*}
	Under the given assumptions $x,v_1>0$, $|z_1|\not=|z_2|$, and $v_2\not= v_1$, this term is strictly positive for $v_2>0$, if $v_2<v_1$ and strictly negative if $v_2>v_1$. This is the statement of the lemma.
	
	\section{Discussion of Assumption \ref{ass:SmatrixInvertibleAndSpectralRadius}}\label{app:spectralradius}
	
	The spectral radius $\rho(\nv S)$ for non-negative $\nv S$ is bounded from above by (cf. \cite{capponi2015price} and \cite{horn1985matrix}, Corollary 8.1.29): 
	\begin{equation*}
		\rho(\nv S)\le \underset{k=1,\ldots,K}{\max}  \frac{\sum\nolimits_{i=1}^N \kappa^i\alpha^{ki} \sum\nolimits_{\ell=1}^K Q_0^{\ell i}}{\gamma^k Q^{k,\rm nb}}.
	\end{equation*}
	The spectral radius is small if the size of the nonbanking sector is large in comparison to the size of the leverage targeting banking sector, leverage targets are not too large and price elasticities are not too small.
	
	\section{Approximation Accuracy}\label{app:accuracy}
	The accuracy of the first order approximation to market capitalization depends on the spectral radius of the matrix $\nv S$, as stated in the following lemma. 
	
	\begin{lemma}\label{lem:normaccuracy}
		Assume that the spectral radius $\rho(\nv S)$ is smaller than one, as stated in Assumption \ref{ass:SmatrixInvertibleAndSpectralRadius}. For every $\varepsilon>0$ such that $\rho(\nv S)+\varepsilon<1$, there exists a matrix norm $\|\cdot\|$ such that $\| \nv S \| \le \rho(\nv S)+\varepsilon$ and
		$$
		\|(\nv I-\nv S)^{-1}-(\nv I+\nv S)\| \leq \frac{\|\nv S\|^2}{1-\|\nv S\|} \leq \frac{(\rho(\nv S)+\varepsilon)^2}{1-(\rho(\nv S)+\varepsilon)} .
		$$
		This implies
		$$\rho((\nv I-\nv S)^{-1}-(\nv I+\nv S)) \leq \frac{\rho(\nv S)^2}{1-\rho(\nv S)}.$$
	\end{lemma}
	
	\begin{proof}
		Let $\varepsilon>0$ such that $\rho(\nv S)+\varepsilon<1$. There exists a matrix norm $\|\cdot\|$ such that $\|\nv S\|\le \rho(\nv S)+\varepsilon<1$. It then holds that
		\begin{align*}
			\|(\nv I-\nv S)^{-1}-(\nv I+\nv S)\|=\|\sum_{j=2}^{\infty} \nv S^j\|&\le \sum_{j=2}^{\infty} \|\nv S^j\|\le \sum_{j=2}^{\infty} \|\nv S\|^j=\sum_{j=0}^{\infty} \|\nv S\|^j-(\|\nv I\|+\|\nv S\|)\\
			&=\frac{1}{1-\|\nv S\|}-(1+\|\nv S\|)=\frac{1}{1-\|\nv S\|}-\frac{(1+\|\nv S\|)(1-\|\nv S\|)}{1-\|\nv S\|}\\
			&=\frac{1-(1-\|\nv S\|^2)}{1-\|\nv S\|}=\frac{\|\nv S\|^2}{1-\|\nv S\|}\leq \frac{(\rho(\nv S)+\varepsilon)^2}{1-(\rho(\nv S)+\varepsilon)},
		\end{align*}
		which proves the first claim. The implication follows from, first, observing that the following inequality holds for any matrix norm $$\rho((\nv I-\nv S)^{-1}-(\nv I+\nv S))\leq \|(\nv I-\nv S)^{-1}-(\nv I+\nv S)\|$$ and, second, letting $\varepsilon$ on the right-hand side approach zero.
	\end{proof}
	
	Observe that the function $x^2/(1-x)$ is strictly increasing for $x\in[0,1)$ and converges to zero as $x\rightarrow 0$. Hence, the smaller the spectral radius $\rho(\nv S)$, the smaller the approximation error. If $\rho(\nv S)$ approaches zero, the approximation error converges monotonically to zero.

	{\Red
	\section{Correlated Shocks}\label{app:dependentshocks}
	
	Assumption~\ref{as:assets} is typically not satisfied for primary assets in the market. Nevertheless, starting from these primary assets, one can easily construct portfolios with normalized asset prices that are uncorrelated and span the same space of trading opportunities.  
	
	We denote the primary asset shocks by $\widetilde{Z}_1,\ldots,\widetilde{Z}_K$ with corresponding prices $P_{\widetilde{Z}_1},\ldots,P_{\widetilde{Z}_K}$, and set $\widetilde{Z}=(\widetilde{Z}_1,\ldots,\widetilde{Z}_K)^\top$, $P_{\widetilde{Z}} = (P_{\widetilde{Z}_1},\ldots,P_{\widetilde{Z}_K})^\top$. Since the covariance matrix ${\rm Cov}(\widetilde{Z})$ is a real symmetric matrix, there exists an orthogonal matrix $\nv{T}$ such that $\nv{T} {\rm Cov}(\widetilde{Z}) \nv{T}^\top = \nv D$ is diagonal. 
	
	If we now define a vector of new asset shocks by $U := \nv{T} \widetilde{Z}$, then the corresponding new assets span the same space of portfolios as the primary assets, but the asset shocks are now uncorrelated because ${\rm Cov} (U) =\nv D$. The prices of these assets are the components of the vector $P_U=\nv{T} P_{\widetilde{Z}}$. Finally, we construct assets with normalized prices and uncorrelated shocks $Z= (Z_1, \dots, Z_K)^\top$ via a componentwise normalization
	\[Z:=\frac{U}{P_{U}} =  {\rm Diag}(\tfrac{1}{P_U}) \nv{T}  \widetilde{Z} .\]
	Suppose that $\tilde q\in \bbr^K$ is a vector whose components are holdings in the primary assets with shocks $\tilde Z$. In our context, $\tilde q$ is placeholder for holdings of banks, nonbanks, and total holdings. The corresponding holdings in the newly constructed assets with uncorrelated shocks and normalized prices are equal to
	$$q = {\rm Diag}(\tfrac{1}{P_U}) \nv{T} \tilde q . $$
	
	All results of this paper apply to the new assets with uncorrelated shocks and normalized prices, if both liquidation strategies and illiquidity characteristics are given in terms of these assets. We also stress that our model assumes that the demand of nonbanks for any individual asset depends on price changes of this asset only, but not on price changes of other assets. In such new setting, one would need to assume that the nonbanking demand is decoupled across the newly constructed uncorrelated and normalized assets. 
}

	\section{Discussion of Assumption \ref{ass:v}}\label{app:ass:v}

		If $v_1=\ldots=v_N$, then it does not matter how each given asset is distributed across the banks because they all have the same systemic significance. As a result, $MSD(\nv Q)$ is constant for all $\nv Q\in\mathcal{D}$, taking into account the constraint $\nv Q\boldsymbol{1}=q$.
		The requirement that systemic significance is not identical across banks is satisfied by any economy, which is not fully homogeneous in terms of targeted leverage and trading strategy. Empirically, \cite{duarte2018fire}, see Table 4 therein, find substantial variation in banks' leverage targets, with a size-weighted average of 13.6, an equal-weighted average of 11.5, and a standard deviation of 3.9.\footnote{Their sample includes the largest 100 banks by assets every quarter, in a sample period from the third quarter of 1999 to the third quarter of 2016 at the quarterly frequency. They also find that $5\%$ and $95\%$ of the leverage target distribution are, respectively, 6.8 and 16.9, and that there is more cross-sectional than time-series variation.} 
		
		If $b^\top v=0$, then there must exist banks in the system which are short some of the assets. Then, the negative price pressure imposed by some banks in the system would be compensated by a positive price pressure created by other banks. In this case, diversification would be f-efficient, and lead to zero deviation of asset prices from fundamental values, i.e., $MSD(\nv Q^{\rm diversified})=(\frac{1}{T} q b^\top v)^\top \nv G (\frac{1}{T} q b^\top v)=0$. 
		In practice, however, $b^\top v>0$  because the budget and the systemic significance of any bank in the system are both positive. Banks are long their assets, including consumer loans, agency, non-agency securities, municipal securities, etc., see, again, Table 4 in \cite{duarte2018fire}. 
	
	\section{Distance from Diversification in the Case $N=K=3$}\label{app:NK3}
	
	In this section, we analyze numerically how our findings established in the case $N=K=2$ would change for a larger economy. As shown in Theorem \ref{thm:efficientholdings} b), if the number of banks is $N>2$, f-efficient holdings are no longer unique. In this case, we consider the f-efficient holdings whose Frobenius distance from diversification is minimal.
	\begin{figure}[h!] 
		\centering
		{\includegraphics[width=0.7\textwidth]{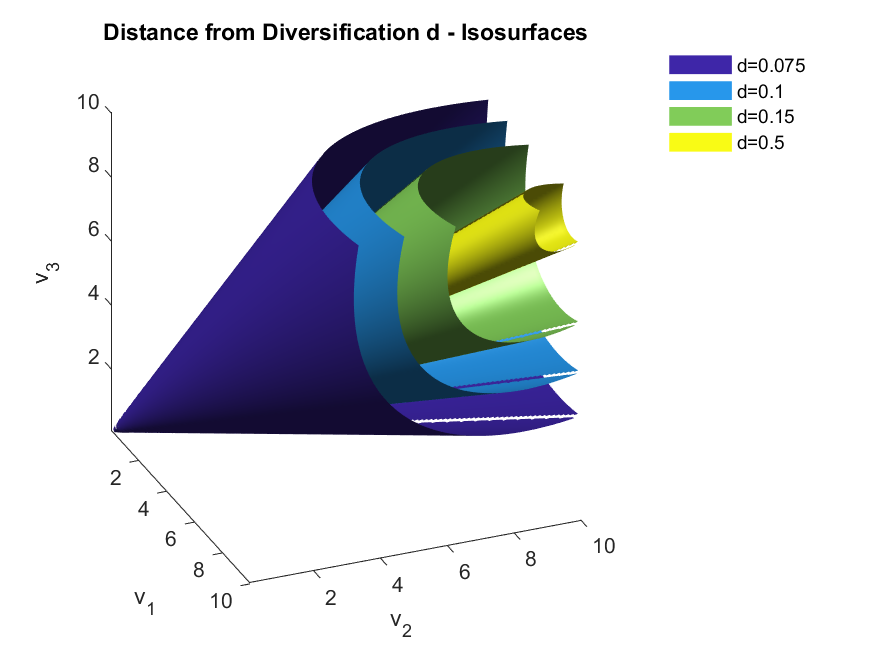}}\\
		\caption{Isosurface plot of the Frobenius distance $d(\nv Q^{\min})$ from diversification for the f-efficient holding matrix $\nv Q^{\min}$ with the smallest distance to diversification. We vary the systemic significance parameters $v_1, v_2$ and $v_3$, and keep fixed shock characteristics, i.e., $\mu=(0.1,0.125,0.15)^\top$ and $\sigma^2=(0.1,0.15625,0.225)^\top$. To ensure comparability with the results of Section \ref{subsec:NK2}, we choose $q_k=b_i=0.08$, $i,k=1,2,3$,  and normalize the total supply of each asset to 1.}
		\label{fig:ExampleN3}
	\end{figure} 
	Figure \ref{fig:ExampleN3} plots the minimal distance of f-efficient holdings from diversification for the case $N=K=3$. We observe that the qualitative findings remain similar to the setting $N=K=2$. The f-efficient holdings get farther away from a full diversification strategy if heterogeneity in banks' systemic significance decreases. The intuition behind the result remains unchanged, i.e., in a system where banks are systemically very close, a full diversification strategy for each bank may lead to larger price pressures because all banks rebalance their portfolios in a similar fashion to meet their leverage targets. 
	
	\section{Non-Uniqueness of Asset Holdings}\label{app:nonuniqueness}
	
	We provide an example to show how the interplay between systemic significance and asset riskiness influences the structure of f-efficient holdings. We also discuss the intuition behind the non-uniqueness of holdings when we move from an economy with $N=2$ banks to one with $N>2$ banks. Consider an economy with $N=3$ banks and $K=3$ assets, where we normalize the total supply of assets within the banking sector and budgets of banks to $q_1=q_2=q_3=b_1=b_2=b_3=x$ with $x:=0.08$. The total supply of each asset is normalized to 1. Asset 1 constitutes the least and asset 3 the most risky asset: $\mu=(0,0,0)^\top$ and $\sigma^2=(0.15,0.2,0.3)^\top$. The three banks are different in their systemic significance parameters $v=(0.15,0.1,0.05)^\top$, i.e., bank 1 is the most and bank 3 the least significant to the system. According to Theorem \ref{thm:efficientholdings} b), f-efficient holdings in this financial system are not unique. Every f-efficient holding matrix is of the form
	\begin{equation}\label{eq:ExNonuniqueSolutions}
	\nv Q^{*1}+\lambda_1\begin{pmatrix}
	1 & -2 & 1\\
	-1 & 2 & -1\\
	0 & 0 & 0
	\end{pmatrix}+\lambda_2\begin{pmatrix}
	1 & -2 & 1\\
	0 & 0 & 0\\
	-1 & 2 & -1
	\end{pmatrix},\quad\lambda_1,\lambda_2\in\mathbb{R},
	\end{equation} 
	where
	\[\nv{Q}^{*1}=x\cdot \begin{pmatrix}
	\frac{2}{3} & \frac{1}{3} & 0\\[0.2em]
	\frac{1}{3} & \frac{1}{3} & \frac{1}{3}\\[0.2em]
	0 & \frac{1}{3} & \frac{2}{3}
	\end{pmatrix}\]
	is a particular solution: the f-efficient holding matrix with the smallest Frobenius distance to fully diversified holdings $\nv{Q}^{\rm diversified}=x\cdot\tfrac{1}{3}\cdot\nv{1}_{3\times 3}$. Hence, $\nv Q^{*1}$ represents a lower bound on how far holdings need to move away from the classical diversification benchmark in order to become f-efficient. Setting $\lambda_1=x\cdot 1/3$ and $\lambda_2=-x\cdot1/6$ in equation \eqref{eq:ExNonuniqueSolutions}, we obtain a second particular solution:
	\[\nv{Q}^{*2}=x\cdot \begin{pmatrix}
	\frac{5}{6} & 0 & \frac{1}{6}\\[0.2em]
	0 & 1 & 0\\[0.2em]
	\frac{1}{6} & 0 & \frac{5}{6}
	\end{pmatrix},\]
	which represents those f-efficient holdings with the smallest distance to a fully diverse holding matrix $\nv Q^{\rm diverse}=x\cdot \nv I_{3}$. Note, first, that this type of diverse holdings are only defined in the case $N=K$. Second, observe that holdings with the smallest distance to diversity do not maximize the distance to diversification.\footnote{Since $\lambda_1,\lambda_2$ in formula \eqref{eq:ExNonuniqueSolutions} are unbounded, the distance to diversification within the set of f-efficient holdings is unbounded.}
	
	The example provides the following insights:
	\begin{itemize}
		\item In both $\nv Q^{*1}$ and $\nv Q^{*2}$, the most risky asset 3 is held in the largest proportion by the least systemically significant bank 3. Conversely, the least risky asset 1 is held in the largest proportion by the most significant bank 1. 
		
		This can be generalized as follows: For every f-efficient holding matrix (see \eqref{eq:ExNonuniqueSolutions}), the holdings of the most significant bank 1 in the least risky asset 1 ($x\cdot 2/3+\lambda_1+\lambda_2$) are larger than the holdings of the least significant bank 3 in this asset ($\lambda_1+\lambda_2$). Conversely, for every f-efficient matrix, the holdings of the least significant bank 3 in the most risky asset 3 ($x\cdot 2/3-\lambda_2$) are larger than bank 1's holdings in this asset ($-\lambda_2$).

		\item Note that the null space in Equation \eqref{eq:ExNonuniqueSolutions} is equivalently written as (cf. Remark \ref{rem:FullSolutionSpace} in E-Companion \ref{app:proofs}):
		\[\lambda_1\begin{pmatrix}
		-(v_2-v_3) & v_1-v_3 & -(v_1-v_2)\\
		v_2-v_3 & -(v_1-v_3) & v_1-v_2\\
		0 & 0 & 0
		\end{pmatrix}+\lambda_2\begin{pmatrix}
		-(v_2-v_3) & v_1-v_3 & -(v_1-v_2)\\
		0 & 0 & 0\\
		v_2-v_3 & -(v_1-v_3) & v_1-v_2
		\end{pmatrix}\]
		for $\lambda_1,\lambda_2\in\mathbb{R}$. Hence, any transfer between two assets in each bank's holdings that is done somewhat proportionally to the differences between the systemic significances (of the other two banks) does not alter the mean-squared deviation. By contrast, if there are only two banks in the system, there would be no \lq\lq other two banks\rq\rq, and thus, there exist no transfers which are neutral with respect to the mean-squared deviation. 
		
	\end{itemize}

	\section{f-Efficient Holdings in the Market Scenarios}\label{app:fefficientmarketholdings}
	
	In this section, we provide the f-efficient holdings and banks' systemic significances for each of the scenarios defined in Section \ref{subsec:marketsetting}.
	\begin{itemize}
		\item In the liquidity scenario (L), the systemic significances of the banks equal  $v_1^L\approx  1.23 < 1.37\approx v_2^L$. The second bank---tracking a higher leverage ratio---is systemically more significant than the first bank. f-efficient holdings are given by
		$$
		\nv{Q}^{*,L}\approx\begin{psmallmatrix}
			-3.79 &  -3.79 &   0.87  &  0.87  &  0.87  &  0.87  &  0.87 &   0.87 &   1.37  &  1.37\\
			3.87  &  3.87 &  -0.79  & -0.79 &  -0.79 &  -0.79 &  -0.79 &  -0.79 &  -1.29 &  -1.29
		\end{psmallmatrix}^\top.$$
		The systemically more significant bank 2 is endowed with a higher number of assets of class 1 (high elasticity, low variance); the least significant bank 1 holds a larger portion of the other assets (smaller elasticity, higher variance).

	\item In scenario (I), bank 2 is still systemically more significant than bank 1, i.e., $v_1^{I}\approx  2.91 < 3.23\approx v_2^{I}$; in comparison to scenario (B), both banks' systemic significances increase due to the increased illiquidity and shock variances of assets from group 3. The f-efficient holdings are given by
	$$\nv{Q}^{*,I}\approx\begin{psmallmatrix}
		-3.87  & -3.87  &  1.07 &   1.07  &  1.07  &  1.07  &  1.07 &   1.07  &  0.87 &   0.87\\
		3.95   & 3.95 &  -0.99 &  -0.99 &  -0.99 &  -0.99 &  -0.99 &  -0.99 &  -0.79 &  -0.79
	\end{psmallmatrix}^\top.$$
	Again, the higher the systemic significance of a bank, the lower its holdings of the safer asset relative to the riskier asset.
	
	\item In scenario (H), bank 2 remains systemically more significant than bank 1, with $v_1^{H}\approx 5.54 < 6.16\approx v_2^{H}$, and the significance parameters are higher than in the two other scenarios. The f-efficient holdings are
	$$\nv{Q}^{*,H}\approx\begin{psmallmatrix}
		-0.41 &  -0.41  &  0.10  &  0.10  &  0.10  &  0.10  &  0.10  &  0.10  &  0.31  &  0.31\\
		0.49  &  0.49  & -0.02 &  -0.02 &  -0.02 &  -0.02 &  -0.02 &  -0.02 &  -0.23  & -0.23
	\end{psmallmatrix}^\top.
	$$

	In this scenario, there is little heterogeneity in the riskiness of the assets, and high heterogeneity in banks' systemic significance. As a result, the f-efficient holdings are more evenly distributed, i.e., closer to full diversification (see also the values of the distances given in Table \ref{tab:statistics}).	
\end{itemize}

	\section{f-Efficient Liquidation Strategies}\label{app:liquidationstrategies}
	
	\subsection{General Derivation}
	We want to minimize the mean squared deviation \eqref{eq:MSE} as a function of 
	$$\nv \alpha =\begin{pmatrix}
	\alpha^1 \;\mid\; \cdots \; \mid \alpha^N
	\end{pmatrix}\in\mathbb{R}^{K\times N}, \;\text{for}\; \alpha^i\in\mathbb{R}^{K}\;\text{with }\boldsymbol{1}_K^\top \alpha^i=1\text{ and }\alpha^i\ge\boldsymbol{0}_K\;\text{for all}\; i=1,\ldots,N.$$
	Thus, each bank $i$ is allowed to choose its own personal liquidation strategy $\alpha^i$ and we do not allow for short-selling. In the following lemma, we rewrite the minimization problem as a function of 
	$${\rm vec}(\nv \alpha):=(\alpha^{1\top},\ldots,\alpha^{N\top})^\top\in\mathbb{R}^{KN},$$ i.e., ${\rm vec}(\nv \alpha)$ denotes the vectorization of the matrix $\nv \alpha$. 
	
	\begin{lemma}\label{lem:liquidationVec}
		Minimizing the mean squared deviation as a function of the liquidation strategy matrix $\nv \alpha$ with $\nv \alpha\ge \boldsymbol{0}$ is equivalent to the following problem:
		\begin{align*}\tag{G}
		&\underset{{\rm vec}(\nv \alpha)\in\mathbb{R}^{KN}}{\min} \quad \tfrac{1}{2}\,{{\rm vec}(\nv \alpha)^\top (\nv C \otimes \tfrac{Q_{\rm tot}}{\gamma\circ Q_0^{\rm nb}}(\tfrac{Q_{\rm tot}}{\gamma\circ Q_0^{\rm nb}})^\top ) {\rm vec}(\nv \alpha) } \\[1em]
		&\text{s.t.}\qquad \begin{pmatrix}
		\boldsymbol{1}_K^\top & & 0\\
		& \ddots &\\
		0 & & \boldsymbol{1}_K^\top\\
		\end{pmatrix} {\rm vec}(\nv \alpha)=\boldsymbol{1}_N,\quad {\rm vec}(\nv \alpha)\ge \boldsymbol{0}_{KN},
		\end{align*}
		where $\nv C=(C^{ij})_{i,j=1,\ldots,N}\in\mathbb{R}^{N\times N}$ with $C^{ij}:=2 (\nv Q{\rm Diag}(\kappa)e^i)^\top (\mu\mu^\top+{\rm Diag}(\sigma^2)) (\nv Q{\rm Diag}(\kappa)e^j)$, $e^i\in\mathbb{R}^{N}$ denotes the $i$'th basis vector (i.e., $e^i_j=1$ for $j=i$ and zero, otherwise), and $\otimes$ denotes the Kronecker product\footnote{For two matrices $\nv A\in\bbr^{M\times N}, \nv B \in\bbr^{P\times R}$, the Kronecker product is defined by multiplying every entry of the matrix $\nv A$ by the entire matrix $\nv B$, i.e., $$\nv A\otimes \nv B:=\begin{pmatrix}
			A^{11} \nv B &\cdots& A^{1N}\nv B\\
			\vdots&\ddots &\vdots\\
			A^{M1}\nv B &\cdots & A^{MN}\nv B\\
			\end{pmatrix}\in\bbr^{MP\times NR}.$$}.
	\end{lemma} 
	\begin{proof}
		It holds that $\nv Q{\rm Diag}(\kappa)\nv \alpha^\top\tfrac{Q_{\rm tot}}{\gamma\circ Q_0^{\rm nb}}=\nv Q{\rm Diag}(\kappa)\cdot(\sum\limits_{i=1}^{N} ({\alpha^i}^\top \tfrac{Q_{\rm tot}}{\gamma\circ Q_0^{\rm nb}}) e^i
		).$ Inserting this expression into formula \eqref{eq:MSE} yields
		\begin{align*}
		& MSD(\alpha^1,\ldots,\alpha^N)=\\
		&\left((\sum\limits_{i=1}^{N} ({\alpha^i}^\top \tfrac{Q_{\rm tot}}{\gamma\circ Q_0^{\rm nb}}) )\nv Q{\rm Diag}(\kappa)e^i\right)^\top (\mu\mu^\top+{\rm Diag}(\sigma^2))\left((\sum\limits_{j=1}^{N} ({\alpha^j}^\top \tfrac{Q_{\rm tot}}{\gamma\circ Q_0^{\rm nb}}) )\nv Q{\rm Diag}(\kappa)e^j\right)\\
		=&\sum\limits_{i=1}^{N}\sum\limits_{j=1}^{N} \left(({\alpha^i}^\top \tfrac{Q_{\rm tot}}{\gamma\circ Q_0^{\rm nb}}) \nv Q{\rm Diag}(\kappa)e^i\right)^\top (\mu\mu^\top+{\rm Diag}(\sigma^2)) \left(({\alpha^j}^\top \tfrac{Q_{\rm tot}}{\gamma\circ Q_0^{\rm nb}}) \nv Q{\rm Diag}(\kappa)e^j\right)\\
		=& \sum\limits_{i=1}^{N}\sum\limits_{j=1}^{N} {\alpha^i}^\top \tfrac{Q_{\rm tot}}{\gamma\circ Q_0^{\rm nb}}(\tfrac{Q_{\rm tot}}{\gamma\circ Q_0^{\rm nb}})^\top {\alpha^j} (\nv Q{\rm Diag}(\kappa)e^i)^\top (\mu\mu^\top+{\rm Diag}(\sigma^2)) (\nv Q{\rm Diag}(\kappa)e^j)\\
		=& \frac{1}{2} \sum\limits_{i=1}^{N}\sum\limits_{j=1}^{N} {\alpha^i}^\top (C^{ij}\cdot\tfrac{Q_{\rm tot}}{\gamma\circ Q_0^{\rm nb}}(\tfrac{Q_{\rm tot}}{\gamma\circ Q_0^{\rm nb}})^\top ){\alpha^j}.
		\end{align*} 
		This proves the formula for the objective function. The linear constraint follows from ${\boldsymbol{1}_K^\top \alpha^i=1}$ for all $i=1,\ldots,N$.
	\end{proof}
	
	The following proposition now provides a locally f-efficient liquidation strategy. Its interpretation is given in Remark \ref{rem:mostliquidstrategy}.
	
	\begin{proposition}\label{prop:liqstrategy}
		Let $m:={\max}_{k\in\{1,\ldots,K\}} \tfrac{\gamma^k Q_0^{k,\rm nb}}{Q^k_{\rm tot}} $ denote the maximum entry in the vector $\tfrac{\gamma\circ Q_0^{\rm nb}}{Q_{\rm tot}}$ and denote by
		$$k_m=\{k\in\{1,\ldots,K\}\mid \tfrac{\gamma^k Q_0^{k,\rm nb}}{Q^k_{\rm tot}}=m\}$$ the corresponding index set with cardinality $\# k_m$. For every fixed ${\nv Q\in\mathbb{R}^{K\times N}}$, a local minimizer of the mean squared deviation as a function of the liquidation strategy matrix is given by the matrix $\nv{\alpha}^*$ which is defined by its columns:
		\[ {\alpha^{ik}}^*:=\begin{cases}
		\frac{1}{\# k_m}, &\text{if $k\in k_m$,}\\
		0, & \text{otherwise,}
		\end{cases} \quad (k=1,\ldots,K,\quad i=1,\ldots,N).\]
	\end{proposition}
	
	\begin{proof}
		It is easily checked that the triplet $({\rm vec}(\nv \alpha^*),\lambda^*,s^*)$, defined as follows, solves the KKT conditions belonging to the optimization problem (G): ${\rm vec}(\nv \alpha^*)=({{\alpha^1}^*}^\top,\ldots,{{\alpha^N}^*}^\top)^\top$ as defined in Proposition \ref{prop:liqstrategy}, $\lambda^*=(\lambda^*_1,\ldots,\lambda_N^*)\in\mathbb{R}^N$, and $s^*=({s^1}^*,\ldots,{s^N}^*)\in\mathbb{R}^{KN}$, where
		$$\lambda_i^*=\tfrac{\sum\nolimits_{j=1}^{N}C^{ij}}{m^2}, \quad {s^i}^*=\tfrac{\sum\nolimits_{j=1}^{N}C^{ij}}{m}(\tfrac{Q_{\rm tot}}{\gamma\circ Q_0^{\rm nb}}-\tfrac{1}{m}\boldsymbol{1}_K),\quad\text{for all}\; i=1,\ldots,N.$$
		Next, we need to check f-efficiency. Let $g_{ik}(\nv\alpha):=-\alpha^{ik}$, and $h_i(\nv\alpha):=\sum\nolimits_{k=1}^K \alpha^{ik}-1$ describe the non-negativity and linear conditions of the optimization problem (G), i.e., the conditions translate into $g_{ik}(\nv\alpha)\le 0$ and $h_i(\nv\alpha)=0$ for all $i=1,\ldots,N$ and $k=1,\ldots,K$. Let $d=(d_1^\top,\ldots,d_N^\top)\in\mathbb{R}^{KN}$, where $d_i\in\mathbb{R}^K$ for all $i=1,\ldots,N$ and define 
		{\small\begin{align*}
			\mathcal{F}(\nv\alpha)&:=\{d\not=\nv 0 \mid \nabla g_{ik}(\nv\alpha)^\top d \begin{cases}
			= 0, \;k\notin k_m,\\
			\le 0, \;k\in k_m,\\
			\end{cases}\hspace{-0.3cm}\text{and} \; \nabla h_i(\nv \alpha)^\top d = 0,\;\forall i\in\{1,\ldots,N\},\;\forall k\in\{1,\ldots,K\} \}.
			\end{align*}}
		Applying the second order sufficiency conditions (cf. Theorem 5.2 in \cite{freund2016optimality}), the KKT point $({\rm vec}(\nv \alpha^*),\lambda^*,s^*)$ constitutes a local minimum if for all $d\in\mathcal{F}(\nv \alpha^*)$ it holds that $d^\top (\nv C\otimes (\tfrac{Q_{\rm tot}}{\gamma\circ Q_0^{\rm nb}})(\tfrac{Q_{\rm tot}}{\gamma\circ Q_0^{\rm nb}})^\top) d>0$. We have
		{\small\begin{align*}
			\mathcal{F}(\nv\alpha)&=\{d\not=\nv 0 \mid d_{ik} \begin{cases}
			= 0,\; k\notin k_m,\\
			\ge 0, \; k\in k_m,
			\end{cases} \hspace{-0.3cm}  \text{and} \; \sum_{k=1}^K d_{ik} = 0, \;\forall i\in\{1,\ldots,N\} \}=\emptyset\\
			\end{align*}}
		for all $\nv\alpha$ and, hence, the second order sufficiency condition is always fulfilled. Thus, $({\rm vec}(\nv \alpha^*),\lambda^*,s^*)$ constitutes a local minimum.
	\end{proof}
	
	\begin{remark}\label{rem:mostliquidstrategy}
		\begin{itemize}
			\item Note that we may characterize liquidity of asset $k$ by its product of elasticity and supply in the nonbanking sector weighted by total supply, i.e., by ${\gamma^k\cdot Q_0^{k,\rm nb}}/{Q^k_{\rm tot}}$. Proposition \ref{prop:liqstrategy} thus shows that the f-efficient liquidation strategy of banks is given by selling solely the most liquid asset. We will refer to this f-efficient strategy as the \textbf{most-liquid-strategy}.
			\item Note that the most-liquid-strategy depends on (the row sums of) the holding matrix $\nv Q$ since it holds that $Q_0^{k,\rm nb}=Q^k_{\rm tot}-\sum\nolimits_{i=1}^{N} Q^{ki}$ for all assets $k=1,\ldots,K$.
			\item In the special case that we ex-ante assume that all banks follow the same liquidation strategy, the most-liquid-strategy even constitutes a global minimizer of the mean squared deviation. This example is analyzed in E-Companion \ref{app:liqstrategyexample}.
		\end{itemize}
	\end{remark}

\subsection{Liquidation Strategy Example}\label{app:liqstrategyexample}

In this case study, we ex-ante assume that all banks act homogeneously in that they liquidate their portfolios in the exact same way. This assumption leads to the following structure of the liquidation strategy matrix:
\[\nv\alpha=\left(\widetilde{\alpha}\; | \cdots |\; \widetilde{\alpha} \right), \]
for a vector $\widetilde{\alpha}\in\mathbb{R}^{K}$, with ${\nv 1}^\top \widetilde{\alpha}=1$, specifying the banks' liquidation of each asset. 
For the mean squared deviation, this leads to the equation
\[ MSD(\widetilde{\alpha},\nv Q)=(\widetilde{\alpha}^\top \tfrac{Q_{\rm tot}}{\gamma\circ Q_0^{\rm 
		nb}})^2\cdot (\nv Q \kappa)^\top (\mu\mu^\top+{\rm Diag}(\sigma^2))(\nv Q \kappa). \]
For a fixed given holding matrix $\nv Q$, we define an \textit{f-efficient bank-independent liquidation strategy} $\widetilde{\alpha}$ as a minimizer of $MSD(\cdot,\nv Q)$ over all $\widetilde{\alpha}\in\mathbb{R}_{\ge 0}^K$ with ${\nv 1}^\top \widetilde{\alpha}=1$. We have the following result.
\begin{proposition}\label{prop:BIalpha}
	For every fixed ${\nv Q\in\mathbb{R}^{K\times N}}$, the most-liquid-strategy constitutes a globally f-efficient bank-independent liquidation strategy.
\end{proposition}
\begin{proof}
	Minimizing $MSD(\widetilde{\alpha},\nv Q)$ for a fixed $\nv Q$ over all $\widetilde{\alpha}\in\mathbb{R}_{\ge 0}^K$ with ${\nv 1}^\top \widetilde{\alpha}=1$ is equivalent to
	\begin{align*}
	&\underset{\widetilde{\alpha}\in\mathbb{R}^K}{\min} \quad \tfrac{1}{2}{\widetilde{\alpha}^\top \tfrac{Q_{\rm tot}}{\gamma\circ Q_0^{\rm nb}}(\tfrac{Q_{\rm tot}}{\gamma\circ Q_0^{\rm nb}})^\top \widetilde{\alpha} } \\\tag{BI}
	&\text{s.t.}\qquad {\nv 1}^\top \widetilde{\alpha}=1,\quad \widetilde{\alpha}\ge \boldsymbol{0}.
	\end{align*}
	The KKT conditions for this optimization problem read
	\begin{align*}
	&{\nv 1}^\top {\widetilde{\alpha}}=1,\quad
	\tfrac{Q_{\rm tot}}{\gamma\circ Q_0^{\rm nb}}(\tfrac{Q_{\rm tot}}{\gamma\circ Q_0^{\rm nb}})^\top  {\widetilde{\alpha}}-\lambda\boldsymbol{1}-s=\boldsymbol{0},\quad
	{\widetilde{\alpha}}\ge \boldsymbol{0}, s\ge \boldsymbol{0},\quad
	{\widetilde{\alpha}}_k s_k=0, \;(k=1,\ldots,K).
	\end{align*}
	Let $m,k_m$ and $\# k_m$ be defined as in Proposition \ref{prop:liqstrategy}. Direct calculation shows that a solution to the KKT conditions is given by $(\widetilde{\alpha}^*,\lambda^*,s^*)$ defined through
	\begin{align*}
	\widetilde{\alpha}^*_k:=\begin{cases}
	\frac{1}{\# k_m}, &\text{if $k\in k_m$,}\\
	0, & \text{otherwise.}\end{cases}, \;(k=1,\ldots,K),\qquad \lambda^*:=\tfrac{1}{m^2}, \qquad s^*:=\tfrac{1}{m}(\tfrac{Q_{\rm tot}}{\gamma\circ Q_0^{\rm nb}}-\tfrac{1}{m}\boldsymbol{1}).
	\end{align*}
	(BI) possesses a convex domain, linear constraints and a convex objective function, because $\tfrac{Q_{\rm tot}}{\gamma\circ Q_0^{\rm nb}}(\tfrac{Q_{\rm tot}}{\gamma\circ Q_0^{\rm nb}})^\top$ is positive semidefinite. Hence, $\widetilde{\alpha}^*$ constitutes a global minimizer of (BI). 
\end{proof}

\end{appendices}

\end{document}